\documentclass[journal]{IEEEtran}
\usepackage{blindtext}

\usepackage{balance}		

\usepackage{cite}

\usepackage[pdftex]{graphicx}

\usepackage[cmex10]{amsmath}
\usepackage{amssymb}
\usepackage{siunitx}
\usepackage[tight,footnotesize]{subfigure}

\usepackage{url}
\usepackage{mathrsfs}
\usepackage{amsthm}
\usepackage{multirow}
\usepackage{amsmath}
\usepackage{cleveref}

\allowdisplaybreaks

\usepackage{xcolor}
\definecolor{morange}{rgb}{0.8,0.2,0}
\definecolor{mblue}{rgb}{0,0.3,1.0}
\definecolor{mgreen}{rgb}{0.2,0.4,0}

\newcommand{\txJ}[1]{\text{Tx}_{#1}}
\newcommand{\rxJ}[1]{\text{Rx}_{#1}}
\newcommand{\yrxJN}[2]{y_{\text{Rx}_{#1}}[#2]}
\newcommand{\yrxJNphi}[2]{y^{\varphi}_{\text{Rx}_{#1}}[#2]}
\newcommand{\yrxJNsys}[2]{y_{\text{Rx}_{#1}}[#2]}

\newcommand{\Ne}[2]{N_{#1}[#2]}
\newcommand{\Nephi}[2]{N^{\varphi}_{#1}[#2]}

\newcommand{\Nesys}[2]{N_{#1}[#2]}
\newcommand{\pe}[1]{P_{e}^{#1}}
\newcommand{\pec}[2]{P_{e,{#1}}^{#2}}
\newcommand{\pk}[2]{P_{#1}[#2]}
\newcommand{\pksys}[2]{P_{#1}[#2]}
\newcommand{\phipksys}[2]{P^{\varphi}_{#1}[#2]}
\newcommand{\Ntx}{N_1}

\newcommand{\pdft}[3]{f_{#1}^{#2}(#3)}
\newcommand{\pdfts}[4]{f_{#1}^{#2}(#3,\!#4\!)}

\newcommand{\pdfsisot}[3]{g^{#2}_{#1}(#3)}
\newcommand{\cdfsisot}[3]{G^{#2}_{#1}(#3)}
\newcommand{\enterp}[2]{s_{{#1},{#2}}}

\newcommand{\fixpnew}[2]{s^{#2}_{#1}}

\newcommand{\cdft}[3]{F_{#1}^{#2}(#3)}

\newcommand{\dptRxJ}[2]{d_{#1}^{\text{Rx}_{#2}}}

\newcommand{\expe}[1]{\mathbb{E}\left[ #1 \right]}

\newcommand{\sminus}{\scalebox{0.75}[1.0]{$\,-\,$}}
\newcommand\numberthis{\addtocounter{equation}{1}\tag{\theequation}}

\theoremstyle{plain}
\newtheorem{theorem}{Theorem}

\newtheorem{remark}{Remark}

\begin{document}

\title{Two-Way Molecular Communications }

\author{Jong Woo Kwak,~\IEEEmembership{Student Member,~IEEE,}
        H. Birkan Yilmaz,~\IEEEmembership{Member,~IEEE,} Nariman Farsad,~\IEEEmembership{Member,~IEEE,} Chan-Byoung Chae,~\IEEEmembership{Senior Member,~IEEE,} 
        and~Andrea Goldsmith,~\IEEEmembership{Fellow,~IEEE}
\thanks{J. Kwak and C.-B. Chae are with the School of Integrated Technology, Yonsei University, Korea. H. B. Yilmaz was with the School of integrated Technology, Yonsei University, Korea, and now he is with the Dept. of Network Engineering, Polytechnic University of Catalonia, Spain. N. Farsad and A. Goldsmith are with the Department of Electrical Engineering, Stanford University, USA.}
\thanks{This research was supported in part by the MSIT (Ministry of Science and ICT), Korea, under the ICT Consilience Creative program (IITP-2017-2017-0-01015), the BA-SIC Science Research Program (2017R1A1A1A05001439) through the National Research Foundation of Korea, the Government of Catalonia's Secretariat for Universities and Research via the Beatriu de Pinos postdoctoral program, and the NSF Center for Science of Information (CSoI) under grant CCF-0939370.}
\thanks{Part of this work is presented in~\cite{kwack_conference}.}

}

\markboth{}%
{Short Title HERE}


\maketitle

\begin{abstract}
	For nano-scale communications, there must be cooperation and simultaneous communication between nano devices. To this end, in this paper, we investigate two-way (a.k.a. bi-directional) molecular communications between nano devices. If different types of molecules are used for the communication links, the two-way system eliminates the need to consider self-interference. However, in many systems, it is not feasible to use a different type of molecule for each communication link. Thus, we propose a two-way molecular communication system that uses a single type of molecule. We derive a channel model for this system and use it to analyze the proposed system's bit error rate, throughput, and self-interference. Moreover, we propose analog- and digital- self-interference cancellation techniques. The enhancement of link-level performance using these techniques is confirmed with both particle-based simulations and analytical results. 

\end{abstract}


%
\IEEEpeerreviewmaketitle

\section{Introduction}

\IEEEPARstart{O}{ver} the past decade, developments in the field of nano robotics have enabled the use of nano devices in various technologies and especially in those used by the bio-medical industry~\cite{lavan2003smallSS,requicha2003nanorobotsNEMS,basu2005syntheticMS,cavalcanti2005nanoroboticsCD,couvreur2006nanotechnologyID}. Well-organized clusters of nano devices can be used for drug delivery applications and artificial immune systems. Each cluster is responsible for a single task, e.g., discovering or destroying of pathogens. Since a nano device can only perform simple tasks, it is important to have a communication system among nano devices. Radio frequency (RF)-based communication is not suitable for nano devices because of physical limitations such as the size of the antenna, which is typically proportional to the wavelength of the electromagnetic (EM) wave in order to maximize efficiency~\cite{farsad2016comprehensiveSO,guo2015molecularVE}. Furthermore, EM waves---especially at high frequencies---do not propagate well in the body~\cite{malak2012molecularCN}. 

Thus, researchers have focused on molecular communication as an alternative to RF-based communication, where information is transmitted via molecules. One such system is that of molecular communication via diffusion (MCvD). Here molecules are propagated in an environment by diffusion~\cite{birkan2010energy}. An MCvD system mainly consists of the following: transmitter nodes capable of emitting and modulating information through molecules, receiver nodes capable of receiving and demodulating molecular signals, information molecules to transfer information, and a fluid environment to host nodes and molecules. Physical realization of the transceiver is also an important issue. The authors in~\cite{nature1, nature2} used Quorum Sensing~\cite{QS} to synchronize the activities of a cluster of engineered bacteria implemented on a chip. Using this orchestrated bacteria cluster as a node, the authors in~\cite{MCN} proposed a diffusion-based molecular communication network.

One of the main challenges in MCvD is to establish channel models for representing the molecular received signal (i.e., a time-dependent solution of the fraction of received molecules by time $t$). Some known channel models assume that the arrival time of molecules are a first-passage time process (i.e., information molecules are absorbed whenever they hit a receiver)~\cite{redner2001guideTF,srinivas2012molecularCI,nakano2012channelMA,yilmaz2014threeDC,genc2016isiAM}. 
Compared to the case of transparent receivers, using the absorbing receiver makes it hard to derive the channel model due to the additional boundary condition. In this paper, we consider the absorbing receivers.
The authors in~\cite{yilmaz2014threeDC} modeled the molecular received signal in a three-dimensional (3-D) environment-- a point source represented a transmitter, and an absorbing sphere represented a receiver. For this basic topology, it is possible to acquire an analytical closed form of the channel model representing the molecular received signal due to spherical symmetry. If the system has more than one absorbing sphere (receiver), however, such symmetry disappears. It is then difficult to model the arrival times mathematically. 

Other challenges in MCvD include low transmission rates due to severe inter-symbol interference (ISI). In MCvD, ISI occurs when the molecules of a previous symbol are absorbed by the target receiver in the current or future symbol slots. The heavy tail nature of impulse responses in MCvD causes severe ISI. Thus, several researchers \cite{adam2014enzyme,tepe2015ISI} have suggested ISI mitigation techniques, including enzymatic degradation of ISI using different molecule types. All of these works assumed a one-way MCvD system whereby molecules are transmitted in one direction from the transmitter to the receiver. While this simplifies the design, recent work in full-duplex radio communications~\cite{stanford2014, passivecancellation,protominkeun} indicates that data rate gains along with other performance advantages may be obtained from two-way communication.

In this paper, we propose a two-way MCvD system that uses a \emph{single} type of molecule for simultaneous communication between two nano devices. If each of the directional communication links uses a different type of molecule, then self-interference\footnote{Self-interference in MCvD refers to a phenomenon in which a molecule emitted from the transmitter is absorbed by its own rather than the intended receiver.} (SI) will not occur~\cite{twoway2017blackseacom, yu2019twoway,twoway_transparent}. However, this is not a feasible solution. First, the nano devices are too simple to perform complex tasks and, second, the number of molecule types will increase rapidly, on the order of $\binom{\ell}{2}$ where $\ell$ is the number of communication links. Therefore, we propose a two-way MCvD system that uses a single type of molecule for each link. Two communication models—a half-duplex
system and a full-duplex system—are considered. In
the half-duplex system, each paired transceiver (i.e., for $\rxJ{1}$, the paired transmitter is $\txJ{2}$) operates
alternately with respect to time. In the full-duplex system,
each paired transceiver operates simultaneously with respect
to time (i.e., it is not necessary to wait until the
other transceiver ends its communication). However, the analytical modeling of such a system is not without difficulty---two absorbing spheres are to use a single type of molecule. Unfortunately, to implement simultaneous communication between two nano devices, we cannot use previous studies on MCvD channel models as such studies assumed a one-way MCvD system with a single absorbing sphere. Therefore, in this paper, we analytically model the molecular received signal in the case of two absorbing spheres, propose SI cancellation (SIC) techniques, and analyze the proposed system's performance in terms of bit error rate (BER) and throughput. The main contributions of this paper are as follows:
\begin{itemize}
	\item The channel model of a two-way MCvD system with two absorbing sphere receivers is derived. To the best of our knowledge, it is the first work that derives the channel model for multi absorbing receiver case. We improve the channel model derived in~\cite{kwack_conference}. More precisely, we analytically derive the asymptotic capture probability (i.e., the probability that the molecules eventually being absorbed by the receivers). By using the derived asymptotic capture probability, we also derive an approximation of the time-dependent solution of the fraction of the received molecules by time $t$.
	The derived channel model is validated through particle-based simulations.
	\item The paper derives the BER expression for two two-way MCvD systems---one with a half-duplex system and one with a full-duplex system. The theoretical BERs are then validated through particle-based simulations. 
	\item The paper proposes two SIC techniques---analog SI cancellation (A-SIC) and digital SI cancellation (D-SIC)---as analytical results show that it is impossible to achieve reliable data transmission without SIC in the full-duplex system. Using the derived channel model and the BER expression, we investigate the optimal values for the normalized detection threshold and the discarding time (i.e., initial part of the received signal is ignored by discarding time $T_c$) in order to minimize the BER. The throughputs and BERs of both systems are analyzed. We compare the SIC-adopted full-duplex system with binary concentration shift keying (BCSK) to
	the half-duplex system with quadrature concentration shift keying (QCSK), which is not provided in~\cite{kwack_conference}.
\end{itemize}

The rest of this paper is organized as follows. In Section~\ref{sec_one_way_comm}, we discuss the conventional one-way MCvD model. In Section~\ref{sec_two_way_comm}, we introduce the proposed two-way MCvD system and SIC techniques. In Section~\ref{sec_ber_formulations}, we investigate the channel model of a two-way MCvD system with two absorbing sphere receivers. Then we present the channel model verifications and BER formulations for the two considered systems. In Section~\ref{sec_numerical_results}, we present performance analysis results in terms of BER and throughput for the two considered systems. Finally, in Section~\ref{sec_conclusion}, we present our conclusions.

\section{One-Way Molecular Communication}
\label{sec_one_way_comm}
We start by providing details of the conventional one-way MCvD system. Also, we present details of the single absorbing receiver channel model so that the reader may understand the differences and challenges with respect to the corresponding model for two-way molecular communication system.

\subsection{Conventional one-way MCvD system}
The conventional one-way MCvD system consists of a point source (point transmitter) and an absorbing sphere (receiver).
In Fig.~\ref{fig_siso_system_model}, the point source (point transmitter), Tx, is separated from the absorbing sphere (receiver), $\rxJ{1}$, the radius of which is denoted by $r_r$, by a distance of $d$. Here we focus on three molecular processes: emission, propagation, and reception. The emission process is related to the modulation of the data bits onto the physical properties of the molecules or the emission time\cite{modul2013isomer}. The propagation process is governed by diffusion and flow\cite{narae2014ISIopt,narae2014nonlinear,nariman2013tabletop}. In this paper, we only consider diffusion. The reception process is related to the acquisition of the molecules at the receiver and the demodulation of the data bits.
\begin{figure}[!t]
	\begin{center}
		\includegraphics[width=0.98\columnwidth,keepaspectratio]%
		{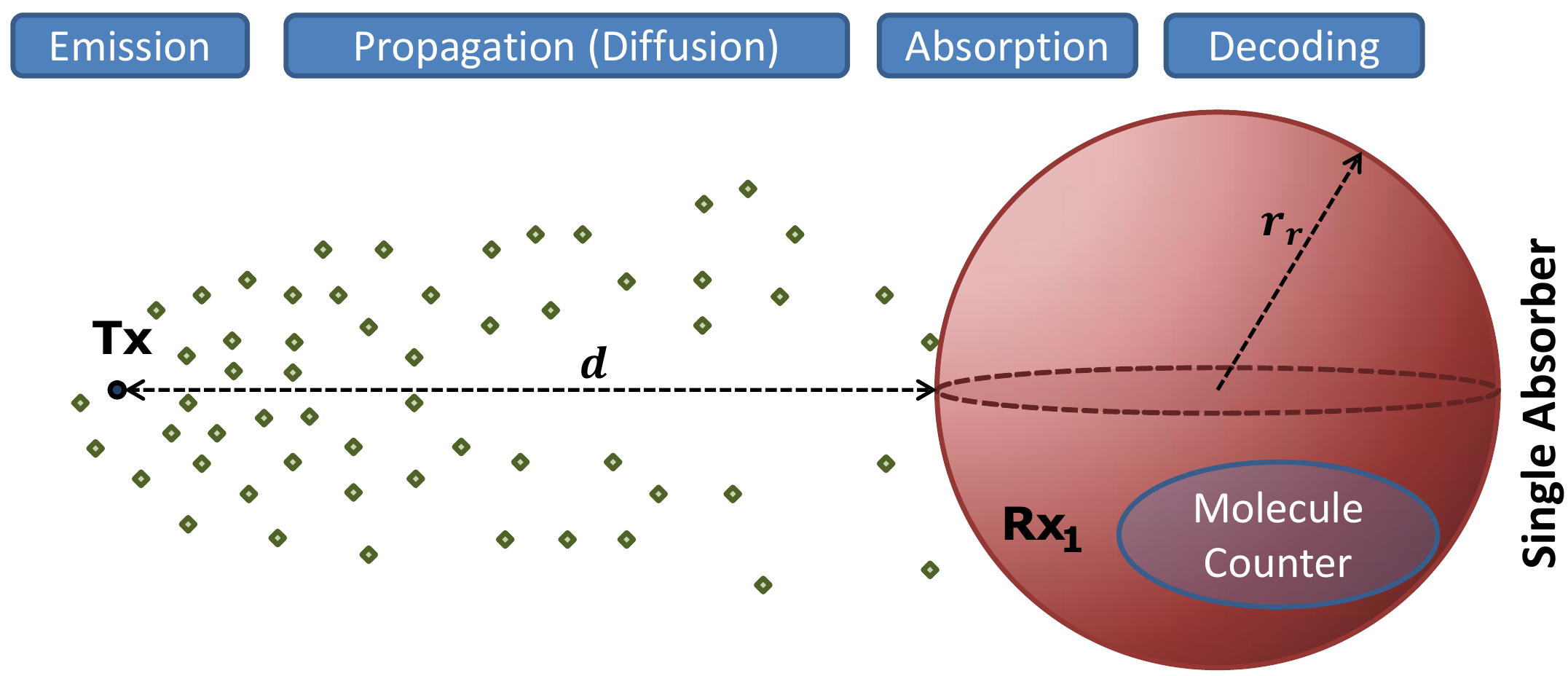}
		\caption{Conventional one-way MCvD transceiver model and the processes with a point transmitter (Tx) and an absorbing sphere ($\rxJ{1}$). }
		\label{fig_siso_system_model}
	\end{center}
\end{figure}

Regarding the propagation process, the interactions between diffusing molecules are ignored since the messenger molecules are assumed to be chemically stable. We assume that the transmitter and the receiver are fully synchronized, which can be achieved by the method introduced in~\cite{shahmohammadian2013blindSI}.

\subsection{Channel Model for One-way Molecular Communication with a Single Receiver }
In diffusion-based systems, propagation of the molecules is governed by Fick's second law in a 3-D environment; that is 
\begin{align}
\label{eq_ficks3}
\frac{\partial p(r,t| r_0)}{\partial t} = D \nabla^2 p(r,t| r_0),
\end{align}
where $\nabla^2$, $p(r,t| r_0)$, and $D$ are the Laplacian operator, the molecule distribution function at time $t$ and distance $r$ given the initial distance $r_0$, and the diffusion constant, respectively. The value of $D$ depends on the temperature, viscosity of the fluid, and the Stokes radius of the molecule~\cite{diffusion_coeffi}.
\begin{figure*}[!t]
	\begin{center}	
		\includegraphics[width=1.8\columnwidth,keepaspectratio]%
		{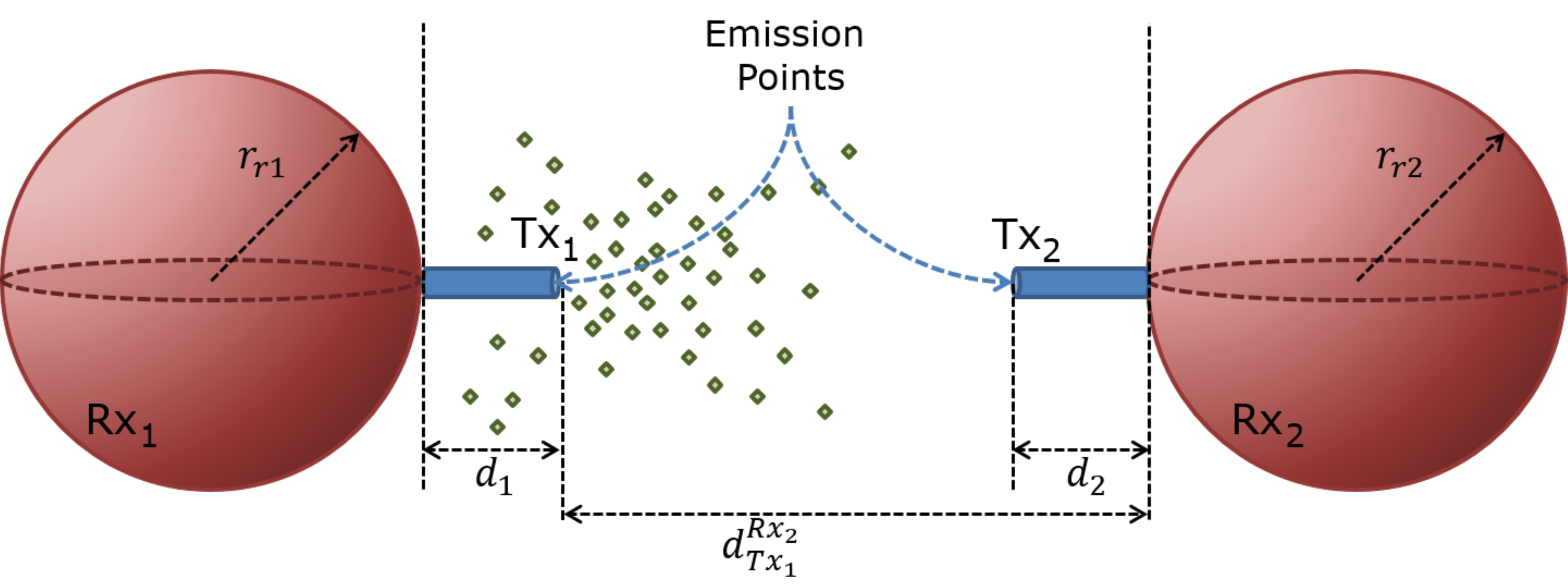}
		\caption{Model for two-way MCvD system featuring two transmitters (i.e., $\txJ{1}$ and $\txJ{2}$) and two receivers (i.e., $\rxJ{1}$ and $\rxJ{2}$).}
		\label{fig_system_model}
	\end{center}
\end{figure*}

Fig.~\ref{fig_siso_system_model} illustrates a simple topology of one-way MCvD. In~\cite{yilmaz2014threeDC}, the expected channel response of one-way MCvD is presented and analyzed from a channel characteristics perspective. Also, a time-dependent solution for a fraction of molecules hitting a single absorbing sphere ($\rxJ{1}$) by time $t$ is presented, as follows:
\begin{align}
\begin{split}
\cdfsisot{1}{\text{Tx}}{t} &=  \int_0^t \pdfsisot{1}{\text{Tx}}{t'} \,dt' \\
        & = \int_0^t \frac{r_r}{r_r\!+\!d} \, \frac{d}{\sqrt{4\pi D t'^3 \,}} e^{-d^2/4Dt'} \,dt' \\
		& = \frac{r_r}{r_r\!+\!d} \,\text{erfc} \left( \frac{d}{\sqrt{4Dt\,}}\right), 
\end{split}
\label{eqn_3d_frac_received_point_src}
\end{align}
where $\text{Tx}$, $\pdfsisot{1}{\text{Tx}}{t}$, $r_r$, $d$, and $\text{erfc}(\cdot)$ represent the location of the point transmitter, the instantaneous hitting probability density (i.e., the arrival time distribution for $\rxJ{1}$ when the molecules are emitted from $\text{Tx}$), the radius of the receiver, the distance from Tx to $\rxJ{1}$, and the complementary error function, respectively.

\section{Two-Way Molecular Communication}
\label{sec_two_way_comm}
In this section, we introduce two modes of operation for the proposed two-way MCvD system, i.e., full-duplex and half-duplex based on time-division. In the former case, severe SI is observed when both transceivers simultaneously modulate signals  using the same type of molecule. Therefore, in this case, it is necessary to apply an SIC technique.

\subsection{Topology}

We consider a 3-D environment with two point sources (transmitters) and two absorbing spheres (receivers). Each transmitter emits molecules without directionality. The propagation of molecules is governed by~\eqref{eq_ficks3}. The molecules are immediately absorbed when they reach the surface of any one of the receivers. Since the absorbed molecules are removed from the system, each molecule is detected at most once.

Fig.~\ref{fig_system_model} shows the model of the proposed system. The blue pipes in Fig.~\ref{fig_system_model} represent the connection between the transmitting part and the receiving part. They are assumed to be transparent to the released molecules. $\txJ{i}$ releases molecules that are intended to be absorbed by $\rxJ{j}$ ($i\neq j$). If the molecules that are released from $\txJ{i}$ are absorbed by $\rxJ{i}$ (not the desired result), then we call this SI. The distance between $\txJ{i}$ and $\rxJ{i}$ is denoted by $d_i$, and the distance from a point $p$ to $\rxJ{j}$ is denoted by $\dptRxJ{p}{j}$. 

The system model illustrated in Fig. 2 is similar to full duplex radios with RF technologies that have antenna separation as passive analog cancellation and digital self interference cancellation. We leave the fabrication issues for the transceivers for our future work and will focus on fundamental analyses of the proposed system.

\subsection{Communication Model \& Modulation}

Consider the following two modes of operation for the proposed two-way MCvD system:
\begin{itemize}
\item Half-duplex system: $\txJ{1}$ and $\txJ{2}$ release molecules alternately (i.e., when $\txJ{1}$ emits molecules, $\rxJ{1}$ and $\rxJ{2}$ receive the molecules but $\rxJ{1}$ does not count the molecules).
\item Full-duplex system: $\txJ{1}$ and $\txJ{2}$ release molecules simultaneously that are intended for $\rxJ{2}$ and $\rxJ{1}$, respectively. Receiver $\rxJ{i}$ receives and counts the molecules that are emitted by $\txJ{i}$ and $\txJ{j}$.
\end{itemize}

In the half-duplex system, at least half of the elements of the bit sequences are not used (i.e., $\txJ{1}$ does not emits molecules when $\txJ{2}$ emits molecules), which is not the case for the full-duplex system. Therefore, the ISI and the SI are much more severe in the full-duplex system. For the $n\text{th}$ symbol period, the molecular received signal is composed of $2n$ bits including the current symbols and the previous $2n\!-\!2$ symbols sent from the two transmitters. The bit sequences for the transmitters are denoted by $x_{1}[1\!:\!n]$ and $x_{2}[1\!:\!n]$. 

For the modulation, we use binary/quadrature concentration shift keying (BCSK, QCSK) ~\cite{modul2013isomer}. We let $\Ntx$ denote the number of molecules for encoding bit-1, and we define that there will be no emission in the case of bit-0 for BCSK. Each of the transmitters has its bit sequences $x_i$ to encode, where $x_i[k]$ denotes the symbol in the $k\text{th}$ symbol duration for $\txJ{i}$.
We define $\pk{ij}{k}$ as the probability that molecules emitted from $\txJ{i}$ hit $\rxJ{j}$ in the $k\text{th}$ symbol duration after the emission, which is formulated as follows:  
\begin{equation}
\pk{ij}{k}\triangleq P_{ij}(kt_s, (k+1)t_s),  \ k\in{\mathbb{N}_{0}},
\label{pij}
\end{equation}
where $t_s$ and ${ P}_{ij}$(${ t}_1$,${ t}_2$) denote the symbol duration and the probability that molecules emitted from $\txJ{i}$ are absorbed by $\rxJ{j}$ but not $\rxJ{i}$ between time $t_1$ and $t_2$ after the emission.
In the rest of this paper we call $\pk{ij}{k}$ the channel coefficient.

In \eqref{pij}, $\pk{ij}{0}$ corresponds to the probability of being absorbed in the current symbol slot. We let $\yrxJN{j}{n}$ denote the number of molecules that are absorbed by $\rxJ{j}$ in the $n\text{th}$ symbol slot. Note that $\yrxJN{j}{n}$ can be affected by the number of molecules released from (i) a pair source at the current symbol slot, (ii) a pair source at the previous time slots, (iii) a non pair source at the current symbol slot, and (iv) non pair source at the previous time slots. 
To formulate $\yrxJN{j}{n}$, we define $\Ne{ij}{k}$ as follows:
\begin{equation}
\Ne{ij}{k}\sim \mathscr{B}(\Ntx,\pk{ij}{k}),
\end{equation}
where $\mathscr{B}(m,p)$ is a binomial distribution with $m$ trials and success probability $p$. Then, $\yrxJN{j}{n}$ can be formulated as follows:
\begin{equation}
\yrxJN{j}{n}\!\triangleq\!\sum_{k=0}^{n-1}(\Ne{ij}{k} \cdot x_i[n-k]+\underbrace{\Ne{jj}{k}\! \cdot\! x_j[n-k]}_{\text{self-interference}})+s[n].
\end{equation} 
To consider the misoperations of the receiver, we add the noise term which is assumed to be a Gaussian random variable, $s[n] \sim \mathscr{N}(0, \sigma_{\text{noise}}^2)$. 
For the sake of tractability, we approximate the binomial distribution as follows:
\begin{equation}
\Nesys{ij}{k} \approx \mathscr{N}(\Ntx \pksys{ij}{k}, \ \Ntx \pksys{ij}{k}(1-\pksys{ij}{k})),
\label{gaussian_approx}
\end{equation}
where $\mathscr{N}(\mu, \sigma^2)$ represents a Gaussian distribution with mean $\mu$ and variance $\sigma^2$. Hence, $\yrxJNsys{j}{n}$ can be expressed as follows (i.e., as a Gaussian random variable, where mean and variance values are dependent upon transmitted bit sequences):
\begin{align}
\begin{split}
\yrxJNsys{j}{n} & \sim\mathscr{N}(\mu_{n,\text{total}},\sigma^2_{n,\text{total}}) \\
\mu_{n,\text{total}} & =\sum_{k=0}^{n-1}\Ntx(\pksys{ij}{k}x_i[n-k]+\pksys{jj}{k}x_j[n-k]) \\
\sigma^2_{n,\text{total}} & =\sigma_{\text{noise}}^2+\sum_{k=0}^{n-1}\Ntx(\pksys{ij}{k}(1-\pksys{ij}{k})x_i[n-k]  \\ 
                  &+\pksys{jj}{k}(1-\pksys{jj}{k})x_j[n-k]).
\label{yrxjn}
\end{split}
\end{align}

\begin{figure}[!t]
	\begin{center}
		\includegraphics[width=0.98\columnwidth,keepaspectratio]
		{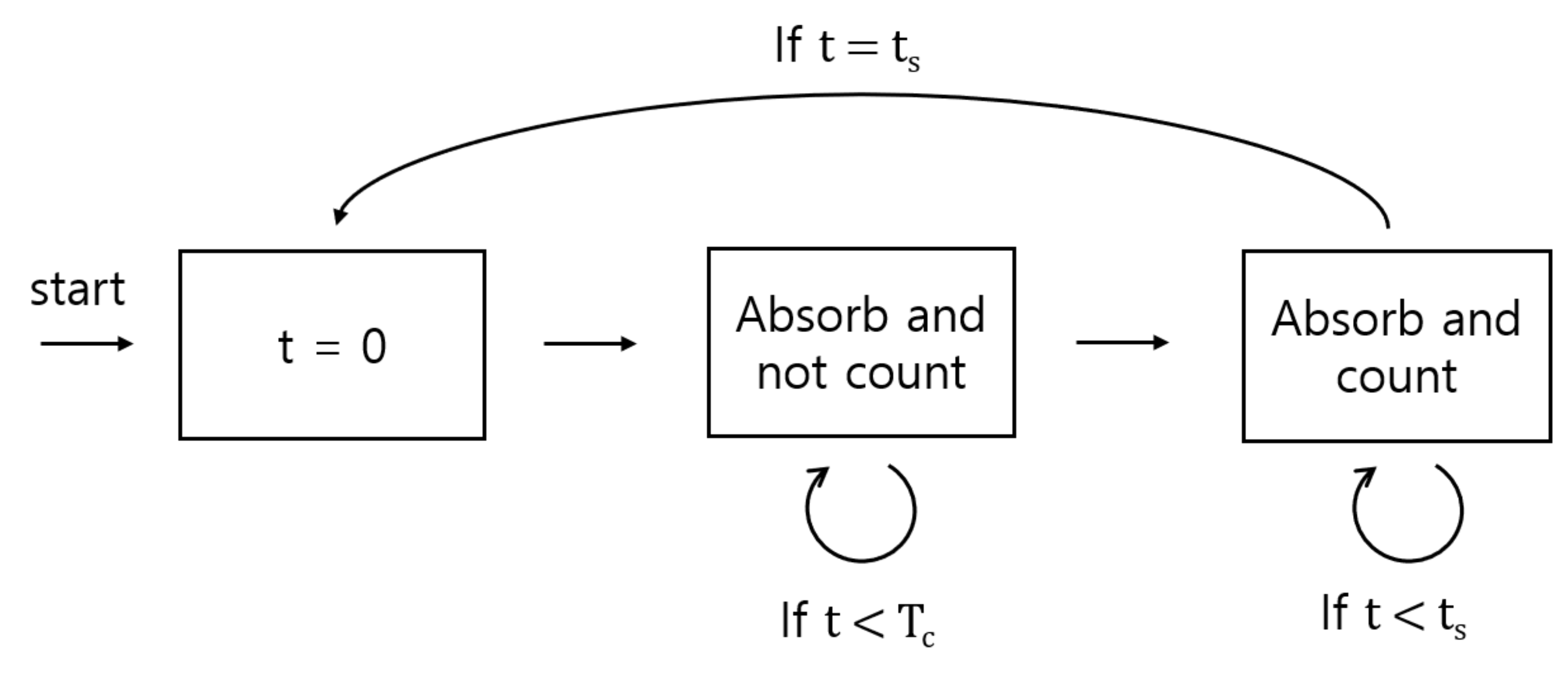}
		\caption{State diagram of the process of A-SIC. For each symbol slot, the absorbed molecules are discarded by $t=T_c$.  }
		\label{fig_A-SIC_introduce}
	\end{center}
\end{figure}
\subsection{Self-Interference Cancellation}

Since the proposed two-way MCvD system comprises two transceivers that use the same type of molecule, the system's receivers are unable to distinguish molecules in terms of the transmitting source. For example, if $\txJ{1}$ sends bit-1 and $\txJ{2}$ sends bit-0, the molecules are released only from $\txJ{1}$. However, those molecules can also be absorbed by $\rxJ{1}$, which is not desired. Then $\rxJ{1}$ may decode the received signal as bit-1, even though its paired transmitter $\txJ{2}$ sends bit-0. In fact, most of the molecules released from $\txJ{1}$ will be absorbed by $\rxJ{1}$ because $\txJ{1}$ is much closer to $\rxJ{1}$ than $\rxJ{2}$. Hence, in this case, the number of received molecules is mostly dependent on the transmitted symbol from the unpaired transmitter, which makes for infeasible communication. Therefore, we propose the following two SIC techniques: 
\begin{itemize}
\item Analog self-interference cancellation (A-SIC): the initial part (i.e., between time 0 and $T_c$) of the molecular received signal for each symbol slot is ignored (see Fig.~\ref{fig_A-SIC_introduce} for the state diagram). In the rest of this paper, we call $T_c$ as discarding time.  
\item Digital self-interference cancellation (D-SIC): we predict the number of SI molecules (i.e., the number of absorbed molecules originating from the unpaired transmitter) from the current bit and subtract it from the molecular received signal
\end{itemize}

\begin{figure*}[!t]
	\begin{center}
		\includegraphics[width=1.8\columnwidth,keepaspectratio]
		{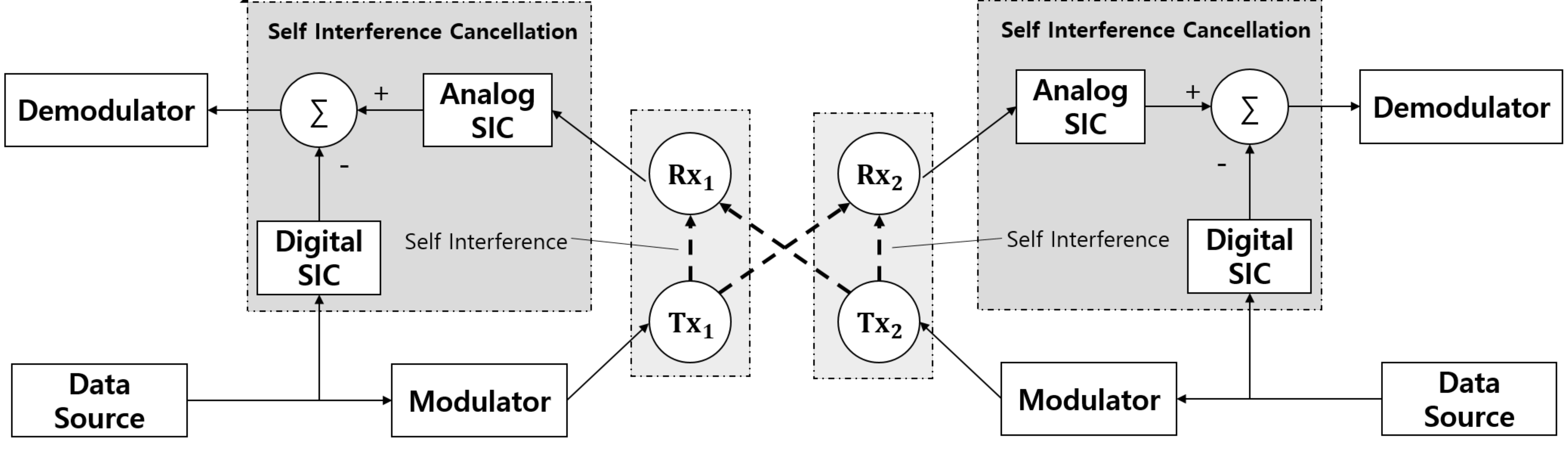}
		\caption{Block diagram of the proposed two-way MCvD system with A-SIC and D-SIC.}
		\label{fig_sic_introduce}
	\end{center}
\end{figure*}
Fig.~\ref{fig_sic_introduce} shows the full-duplex system with A-SIC and D-SIC. Note that the two SIC techniques can be applied separately. 
The channel coefficient of the system with A-SIC, $\phipksys{ij}{k}$, is given as follows:
\begin{align}
\phipksys{ij}{k}=P_{ij}(kt_s+T_c,(k+1)t_s).
\end{align}
Hence, $\Ne{ij}{k}$ of the proposed two-way MCvD system with A-SIC is denoted by $\Nephi{ij}{k}$ and becomes
\begin{align}
 \Nephi{ij}{k} \sim \mathscr{B}(\Ntx, \phipksys{ij}{k}).
 \label{phiN}
\end{align}
Furthermore, $\yrxJNsys{j}{n}$ with A-SIC and D-SIC becomes
\begin{align}
\begin{split}
\yrxJNphi{j}{n} & \triangleq \sum_{k=0}^{n-1} (\Nephi{ij}{k} \cdot x_i[n\!-\!k] + \underbrace{\Nephi{jj}{k} \cdot x_j[n\!-\!k])}_{\text{self-interference}} \\
                  & - \underbrace{\expe{\Nephi{jj}{0} \cdot x_j[n]}}_{\text{when D-SIC applied}}, \\
\text{where} &\ \ \ \expe{\Nephi{jj}{0} \cdot x_j[n]} =\Ntx\phipksys{jj}{0}x_j[n].
\end{split}
\label{eq_D-SIC}
\end{align}
$\expe{\cdot}$ is the expectation operation. 
By applying these two SIC techniques,~\eqref{yrxjn} is transformed into
\begin{align}
\begin{split}
\yrxJNphi{j}{n} & \sim\mathscr{N}(\mu^{\varphi}_{n,\text{total}},(\sigma^{\varphi}_{n,\text{total}})^2) \\
\mu^{\varphi}_{n,\text{total}} & =\sum_{k=0}^{n-1}\Ntx(\phipksys{ij}{k}x_i[n-k]+\phipksys{jj}{k}x_j[n-k]) \\
&-\underbrace{\Ntx\phipksys{jj}{0}x_j[n]}_{\text{when D-SIC applied}}\\
(\sigma^{\varphi}_{n,\text{total}})^2 & =\sigma_{\text{noise}}^2+\sum_{k=0}^{n-1}\Ntx(\phipksys{ij}{k}(1-\phipksys{ij}{k})x_i[n-k]  \\ 
&+\phipksys{jj}{k}(1-\phipksys{jj}{k})x_j[n-k]).
\label{phiyrxjn}
\end{split}
\end{align}

\section{Channel Model \& BER Formulation of Two-Way Molecular Communication}
\label{sec_ber_formulations}
We formulate the BER as a function of detection threshold $\tau$, the number of molecules for encoding bit-1 ($\Ntx$), $\phipksys{ij}{k}$ ($\phipksys{ij}{k}$ when SIC is applied), and the symbol duration $t_s$. In this section, we first derive the channel model to obtain the channel coefficients (i.e., $\pksys{ij}{k}$ and $\phipksys{ij}{k}$) and then utilize these channel coefficients in the BER calculations.

\subsection{Channel Model}

As we mentioned, the channel model is a time-dependent solution of the fraction of received molecules by time $t$. In the proposed two-way MCvD system, we cannot use~\eqref{eqn_3d_frac_received_point_src} directly to derive the channel model since we have to consider the events of molecules being absorbed by each $\rxJ{i}$, which are not independent of each other. The proposed BER formula and the SIC techniques are based on the channel model. However, there is no analytical closed-form solution in the literature for the case of two absorbing spherical receivers. In prior work on molecular MIMO \cite{bonhong2016MIMO,ml2017SPAWC},  researchers considered two pairs of point transmitters and fully absorbing receivers. To obtain the channel models, the authors in \cite{bonhong2016MIMO} and \cite{ml2017SPAWC} utilized a one-way MCvD system channel model for a single receiver \cite{yilmaz2014threeDC} and fitted the coefficients accordingly. In this paper, we propose a new approach to analytically derive the multi-receiver channel model.

Fig.~\ref{fig_topology_general_coordinate} describes the notations and the cartesian coordinates in the derivation. The centers of two receivers are located on $z$-axis and the distance between them is denoted by $\ell$ (i.e., $(0,0,a)$ and $(0,0,a-\ell)$). Our goal is to obtain the fraction of received molecules by time $t$ for the receivers when $\txJ{1}$ is the emitter (i.e., $\cdft{1}{\txJ{1}}{t}$ and $\cdft{2}{\txJ{1}}{t}$ for the receivers $\rxJ{1}$ and $\rxJ{2}$, respectively).

In  the rest of this subsection,  we first derive a closed form of $\lim_{t\to\infty}\cdft{1}{\txJ{1}}{t}$ and $\lim_{t\to\infty}\cdft{2}{\txJ{1}}{t}$, which correspond to the probability that the molecule eventually being absorbed by $\rxJ{1}$ and $\rxJ{2}$, respectively. This so-called asymptotic capture probability was studied in~\cite{smolu_solution}. The authors in~\cite{smolu_solution} derive an exact solution of the asymptotic capture probability for the case of two absorbing spheres, which has an identical radius. In this paper we generalize this solution to release the equal-radius constraint since our system adopts a general topology.

\begin{figure}[!t]
	\begin{center}
		\includegraphics[width=0.98\columnwidth,keepaspectratio]%
		{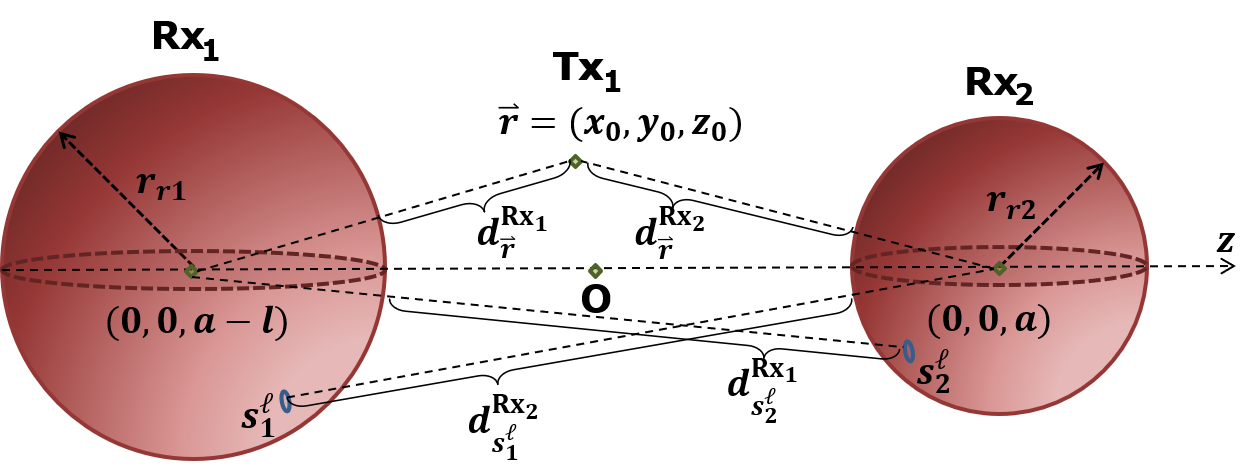}
		
		\caption{Assumed topology in the derivation}
		\label{fig_topology_general_coordinate}
	\end{center}
\end{figure}

In the derivation of the asymptotic capture probability, we use bispherical coordinate $(\mu,\eta,\phi)$~\cite[p. 
1298]{smolu_laplace_solution}). 
Note that bispherical coordinate is obtained by rotating bipolar axes about the line between the two foci, $(0,0,\pm f)$ in cartesian coordinates:
\begin{align}
\begin{split}
\mu&= \tanh^{-1}\left(\frac{f z}{x^2+y^2+z^2+(f/2)^2}\right), \\
\eta&= \tan^{-1}\left(\frac{f(x^2+y^2)^{1/2}}{x^2+y^2+z^2-(f/2)^2}\right), \\
\phi&= \tan^{-1}\left(\frac{y}{x}\right). \\
\end{split}
\label{eq_bispherical}
\end{align}

\begin{figure}[!t]
	\begin{center}
		\includegraphics[width=0.98\columnwidth,keepaspectratio]%
		{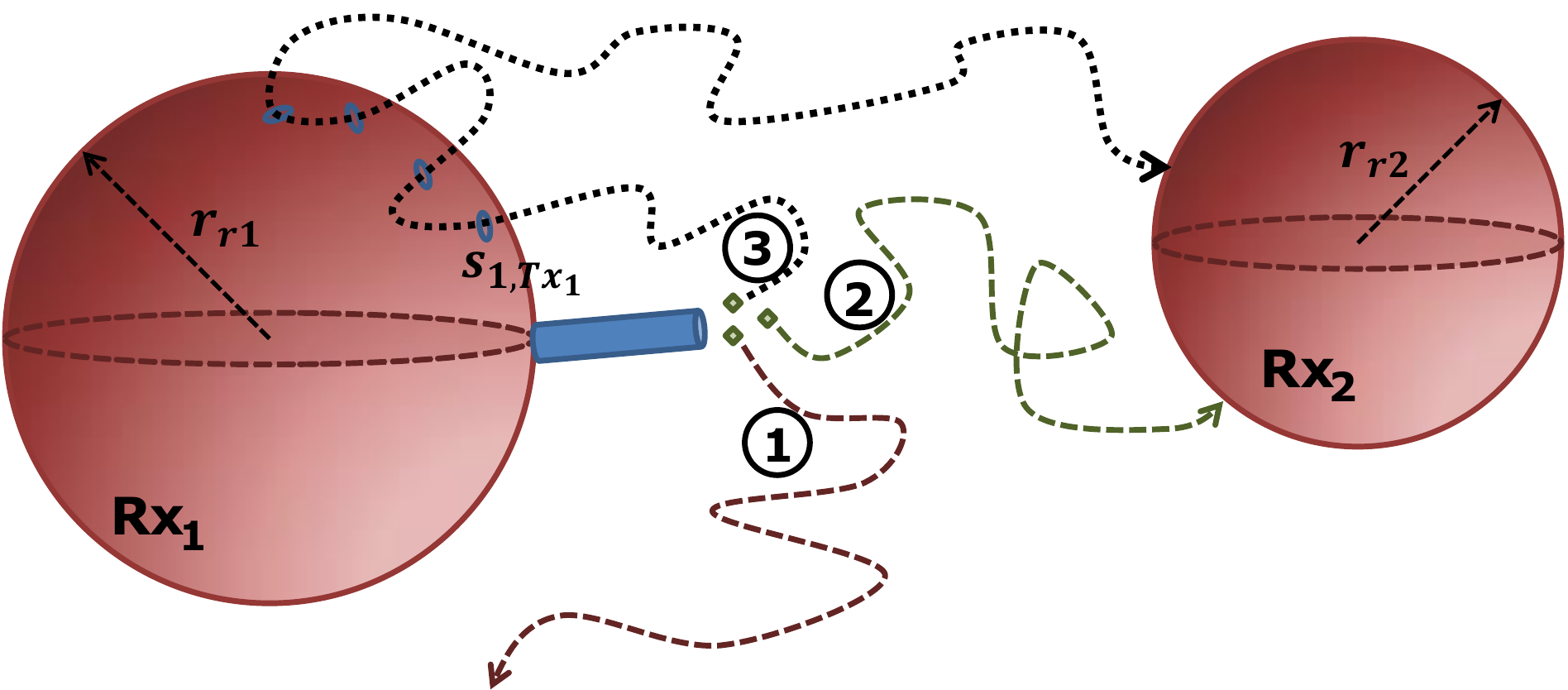}
		
		\caption{Three possible traversal paths for a molecule emitted from $\txJ{1}$. Path~1 corresponds to a molecule not hitting any of the receivers by time $t$. Path~2 corresponds to a molecule that is hitting the destination ($\rxJ{2}$) at time $t$ without hitting $\rxJ{1}$. Path~3 is a virtual path that corresponds to a molecule that is, actually hitting $\rxJ{1}$ but that would hit $\rxJ{2}$ at time $t$ if $\rxJ{1}$ was not in the environment or it was transparent to molecules.}
		\label{fig_topology_general}
	\end{center}
\end{figure}

\begin{theorem}
\label{theorem1}
The asymptotic capture probabilities in the cases of two receivers, i.e., $\lim_{t\to\infty}\cdft{1}{\txJ{1}}{t}$ and $\lim_{t\to\infty}\cdft{2}{\txJ{1}}{t}$, are as follws:
\begin{align}
\begin{split}
\lim_{t\to\infty}&\cdft{1}{\txJ{1}}{t} =\sqrt{2(\cosh \mu_0 - \cos\eta_0)}\\
&\times\left(\sum_{m=0}^{\infty}\frac{e^{-(m+\frac{1}{2})\mu_1}\sinh [(m+\frac{1}{2})(\mu_0+\mu_2)]}{\sinh [(m+\frac{1}{2})(\mu_1+\mu_2)]}P_m(\cos\eta_0)\right), \\
\lim_{t\to\infty}&\cdft{2}{\txJ{1}}{t} =\sqrt{2(\cosh \mu_0 - \cos\eta_0)}\\
&\times\left(\sum_{m=0}^{\infty}\frac{e^{-(m+\frac{1}{2})\mu_2}\sinh [(m+\frac{1}{2})(\mu_1-\mu_0)]}{\sinh [(m+\frac{1}{2})(\mu_1+\mu_2)]}P_m(\cos\eta_0)\right), \\
\end{split} 
\label{eq_capture_prob}
\end{align}
where $\mu_1\!\!=\!\!\cosh^{-1}\left(\!\frac{	\ell-a}{r_{r1}}\!\right)$ and $\mu_2\!\!=\!\!\cosh^{-1}\left(\!\frac{a}{r_{r2}}\!\right)$, $P_m$ is the $m$-th order Legendre polynomial, and $(\mu_0,\eta_0,\phi_0)$ is the bispherical representation of the cartesian coordinate of $\txJ{1}$ (i.e.,  $(x_0,y_0,z_0)$) with two foci $(0,0,\pm f)$, 
\begin{equation}
f=r_{r2}\sinh\left(\cosh^{-1}\left(\frac{a}{r_{r2}}\right) \right).
\label{eq_def_f}
\end{equation}

\end{theorem}

\begin{proof}
		
	Let $k_i(x_0,y_0,z_0)$ denotes an asymptotic capture probability for $\rxJ{i}$, where $\vec{r}\!\!=\!\!(x_0,y_0,z_0)$ is the initial position of the molecule. The asymptotic capture probability satisfies the Laplace equation~\cite{smolu_laplace},
	\begin{equation}
	\begin{aligned}
	\nabla^2 k_1(\vec{r})=0,\\
	\nabla^2 k_2(\vec{r})=0.
	\end{aligned}
	\label{eq_laplace}
	\end{equation}
	The boundary conditions in this case are 
	\begin{equation}
	k_1(\vec{r})\!=\!1,\ k_2(\vec{r})\!=\!0  \ \ \text{when} \ \ d_{\vec{r}}^{\rxJ{1}}=0,
	\label{eq_bcond1}
	\end{equation}
	\begin{equation}
	k_1(\vec{r})\!=\!0,\ k_2(\vec{r})\!=\!1\ \ \text{when} \ \ d_{\vec{r}}^{\rxJ{2}}=0,
	\label{eq_bcond2}
	\end{equation}
	\begin{equation}
	k_1(\vec{r}),k_2(\vec{r}) \to 0 \ \ \text{as} \ \ d_{\vec{r}}^{\rxJ{1}}, d_{\vec{r}}^{\rxJ{2}} \to \infty,
	\label{eq_bcond3}
	\end{equation}
	where $d_{\vec{r}}^{\rxJ{1}}$ and $d_{\vec{r}}^{\rxJ{1}}$ is the distance from the initial point $\vec{r}$ to $\rxJ{1}$ and $\rxJ{2}$, respectively. Without loss of generality, we set $a$ as follows:
	\begin{equation}
	\label{eq_def_a}
	a=\frac{\ell^2+r_{r2}^2-r_{r1}^2}{2\ell},
	\end{equation}
	By setting $f$ and $a$ as~\eqref{eq_def_f} and~\eqref{eq_def_a}, the surface of $\rxJ{1}$ and $\rxJ{2}$ are represented as $\mu\!=\!\mu_1$ and $\mu=-\mu_2$, respectively. This enables us to utilize the typical solutions of the Laplace equation~\cite[p. 
	1299]{smolu_laplace_solution} and derive the solution of~\eqref{eq_laplace} as~\eqref{eq_capture_prob}. 
\end{proof}
It can be easily shown that when $r_{r2}\to0$, the asymptotic capture probability for $\rxJ{2}$ (i.e., $\lim_{t\to\infty}\cdft{2}{\txJ{1}}{t}$ in~\eqref{eq_capture_prob}) goes to 0 as $\mu_2\to\infty$. Moreover, by using~\cite[eqs. (10.3.70)]{smolu_laplace_solution}, we get the following: 
\begin{align}
\begin{split}
\lim_{t\to\infty}\cdft{1}{\txJ{1}}{t} =&\sqrt{\frac{\cosh \mu_0 - \cos\eta_0}{\cosh (\mu_0-2\mu_1) - \cos\eta_0}}\\
=&\frac{r_{r1}}{r_{r1}+d_{\vec{r}}^{\rxJ{1}}}, \ \ \ \text{when} \  r_{r2}\to 0, \\
\end{split} 
\label{eq_capture_prob_limit}
\end{align}
which is the asymptotic capture probability for the single absorbing receiver case~\cite{yilmaz2014threeDC}.

Based on the derived asymptotic capture probability, we provide a derivation of the approximation of the channel model.
We investigate the multi-receiver channel model by considering the possible diffusion paths in the case of a two-way MCvD system. In Fig.~\ref{fig_topology_general}, three possible traversal paths (diffusion paths) for a molecule emitted from $\txJ{1}$ are shown.  $\cdft{2}{\txJ{1}}{t}$ can be expressed as follows:
\begin{equation}
\cdft{2}{\txJ{1}}{t} = \int_{0}^{t}\pdft{2}{\txJ{1}}{t'}dt',
\label{eq_fraction_define}
\end{equation}
where $\pdft{2}{\txJ{1}}{t}$ denotes the instantaneous hitting probability density for the molecules following Path~2.
Note that $\txJ{2}$ will not be considered in the derivation because we can obtain the whole case by superposition. 
\begin{remark}
	$\rxJ{1}$ will be regarded as a non-intended receiver on the path to $\rxJ{2}$ and vice versa.
\end{remark}  
\begin{remark}
	To obtain the probability corresponding to Path~2, we subtract the probability corresponding to Path~3 from the channel model of the one-way MCvD with a single point transmitter, which corresponds to the probability of $\{\text{Path~2} \cup \text{Path~3}\}$.
\end{remark}
\noindent Hence, we consider the instantaneous hitting probability densities for $\rxJ{2}$ as follows:
\begin{equation}
\pdft{2}{\txJ{1}}{t'} = \pdfsisot{2}{\txJ{1}}{t'} - \alpha(t'),
\label{eq_alpha}
\end{equation}
where $\alpha(t')$ corresponds to the instantaneous hitting probability densities for the molecules following Path~3. Additionally, $\pdfsisot{2}{\txJ{1}}{t'}$ corresponds to the single receiver case that is given in the literature as~\eqref{eqn_3d_frac_received_point_src}, i.e., the union of Path~2 and Path~3 as if $\rxJ{1}$ is not in the environment. Therefore, if we calculate $\alpha(t')$, then we can calculate $\pdft{2}{\txJ{1}}{t'}$, which in turn will lead us to $\cdft{2}{\txJ{1}}{t'}$.

By definition, each of the molecules that is moving through Path~3 visits $\rxJ{1}$ at least once. Therefore,  we can segment Path~3 into two parts. For those molecules originating from $\txJ{1}$, we denote the first hitting point on the surface of $\rxJ{1}$ as $\enterp{1}{\txJ{1}}$ and the corresponding first hitting time as $\tau$. Note that $\enterp{1}{\txJ{1}}$ can be an arbitrary point on the surface of $\rxJ{1}$ and $\tau$ can be any real value less than $t'$.

The instantaneous hitting probability density for $\enterp{1}{\txJ{1}}$ on the surface of $\rxJ{1}$ (i.e., the arrival time distribution for $\enterp{1}{\txJ{1}}$ when the molecules are emitted from $\text{Tx}$) is denoted by $\pdfts{1}{\txJ{1}}{t'}{\enterp{1}{\txJ{1}}}$. Since $\enterp{1}{\txJ{1}}$ is the first hitting point on the surface of $\rxJ{1}$, it can be regarded as the starting point of the successive path to $\rxJ{2}$. Hence, \eqref{eq_alpha} can be rewritten as 
\begin{align}
\pdft{2}{\txJ{1}}{t'} \!=\! \pdfsisot{2}{\txJ{1}}{t'} \!-\!\!\! \int_{0}^{t'}\!\!\!\int_{\Omega_1}\!\!\pdfts{1}{\txJ{1}}{\tau}{\enterp{1}{\txJ{1}}}\pdfsisot{2}{\enterp{1}{\txJ{1}}}{t'\!-\!\tau} d\enterp{1}{\txJ{1}} d\tau,
\label{eq_alpha2}
\end{align} 
where $\Omega_1$ indicates the points on the surface of $\rxJ{1}$. To obtain $\pdft{2}{\txJ{1}}{t'}$, we also need to consider $\pdft{1}{\txJ{1}}{t'}$ in a similar way. Thus, we have the following:
\begin{align}
\pdft{1}{\txJ{1}}{t'} \!=\! \pdfsisot{1}{\txJ{1}}{t'} \!-\!\!\! \int_{0}^{t'}\!\!\!\int_{\Omega_2}\!\!\pdfts{2}{\txJ{1}}{\tau}{\enterp{2}{\txJ{1}}}\pdfsisot{1}{\enterp{2}{\txJ{1}}}{t'\!-\!\tau} d\enterp{2}{\txJ{1}} d\tau,
\label{eq_alpha3}
\end{align} 
where $\enterp{2}{\txJ{1}}$ is the first hitting point on $\rxJ{2}$ that is analogous to $\enterp{1}{\txJ{1}}$ for $\rxJ{1}$.

When we apply the mean value theorem for integration to the surface integration in \eqref{eq_alpha2}, we get
\begin{equation}
\pdft{2}{\txJ{1}}{t'} = \pdfsisot{2}{\txJ{1}}{t'}-\! \int_{0}^{t'}\!\!\!\!\pdfsisot{2}{\fixpnew{1}{\prime}(t',\tau)}{t'\!-\!\tau}\!\!\int_{\Omega_1}\!\!\!\pdfts{1}{\txJ{1}}{\tau}{\enterp{1}{\txJ{1}}} d\enterp{1}{\txJ{1}} d\tau,
\label{eq_alpha4}
\end{equation}
where $\fixpnew{1}{\prime}(t',\tau)$ is a point $\in\Omega_1$ .
After the surface integration in \eqref{eq_alpha4}, we obtain
\begin{equation}
\pdft{2}{\txJ{1}}{t'}= \pdfsisot{2}{\txJ{1}}{t'} \!-\! \int_{0}^{t'}\pdfsisot{2}{\fixpnew{1}{\prime}(t',\tau)}{t'\!-\!\tau}\pdft{1}{\txJ{1}}{\tau} d\tau.
\label{pdf_final1}
\end{equation}
In the same way, we can express $\pdft{1}{\txJ{1}}{t'}$ as 
\begin{equation}
\pdft{1}{\txJ{1}}{t'}= \pdfsisot{1}{\txJ{1}}{t'} \!-\! \int_{0}^{t'}\pdfsisot{1}{\fixpnew{2}{\prime}(t',\tau)}{t'\!-\!\tau}\pdft{2}{\txJ{1}}{\tau} d\tau,
\label{pdf_final2}
\end{equation} 
where $\fixpnew{2}{\prime}(t',\tau)$ is a point $\in\Omega_2$.

Note that to derive $\fixpnew{i}{\prime}(t',\tau)$ requires the closed forms of $\pdft{i}{\txJ{1}}{t}$, therefore, it is unattainable to solve the simultaneous equations~\eqref{pdf_final1} and~\eqref{pdf_final2}. Therefore, we set  $\fixpnew{i}{\prime}(t',\tau)$ as a constant $\fixpnew{i}{\ell}\!\in\!\Omega_i$, and then derive the approximation of $\pdft{i}{\txJ{1}}{t}$ by solving~\eqref{pdf_final1} and~\eqref{pdf_final2}. We determine the two constants $\fixpnew{1}{\ell}$ and $\fixpnew{2}{\ell}$ by matching the asymptotic values of approximation forms and~\eqref{eq_capture_prob}.    
By substituting $\fixpnew{i}{\prime}(t',\tau)$ with $\fixpnew{i}{\ell}$ and integrating \eqref{pdf_final1} and \eqref{pdf_final2}, we obtain the following:
 
\begin{align}
\begin{split}
\cdft{2}{\txJ{1}}{t}=\cdfsisot{2}{\txJ{1}}{t}-\int_{0}^{t}\pdfsisot{2}{\fixpnew{1}{\ell}}{t'}*\pdft{1}{\txJ{1}}{t'}  dt' \\
\cdft{1}{\txJ{1}}{t}=\cdfsisot{1}{\txJ{1}}{t}-\int_{0}^{t}\pdfsisot{1}{\fixpnew{2}{\ell}}{t'}*\pdft{2}{\txJ{1}}{t'}  dt'.
\end{split}
\label{cdf_1}
\end{align}
Then, we get
\begin{align}
\begin{split}
\cdft{2}{\txJ{1}}{t}=\cdfsisot{2}{\txJ{1}}{t}-\pdfsisot{2}{\fixpnew{1}{\ell}}{t}*\cdft{1}{\txJ{1}}{t} \\
\cdft{1}{\txJ{1}}{t}=\cdfsisot{1}{\txJ{1}}{t}-\pdfsisot{1}{\fixpnew{2}{\ell}}{t}*\cdft{2}{\txJ{1}}{t}.
\end{split}
\label{cdf_1_app}
\end{align}
Next, we take $t\to\infty$ in~\eqref{cdf_1}, which gives us
\begin{align}
\begin{split}
\lim_{t\to\infty}\cdft{2}{\txJ{1}}{t}&=\lim_{t\to\infty}\cdfsisot{2}{\txJ{1}}{t}-\lim_{t\to\infty}\cdfsisot{2}{\fixpnew{1}{\ell}}{t}  \ \lim_{t\to\infty}\cdft{1}{\txJ{1}}{t}\\
\lim_{t\to\infty}\cdft{1}{\txJ{1}}{t}&=\lim_{t\to\infty}\cdfsisot{1}{\txJ{1}}{t}-\lim_{t\to\infty}\cdfsisot{1}{\fixpnew{2}{\ell}}{t}  \ \lim_{t\to\infty}\cdft{2}{\txJ{1}}{t},
\end{split}
\label{cdf_limit}
\end{align}
In~\eqref{cdf_limit}, $\cdfsisot{2}{\fixpnew{1}{\ell}}{t}$ and $\cdfsisot{2}{\fixpnew{1}{\ell}}{t}$ are expressed as follows:    
\begin{align}
\begin{split}
\cdfsisot{2}{\fixpnew{1}{\ell}}{t}&=\frac{r_{r2}}{r_{r2}\!+\!d_{\fixpnew{1}{\ell}}^{\rxJ{2}}} \,\text{erfc} \left( \frac{d_{\fixpnew{1}{\ell}}^{\rxJ{2}}}{\sqrt{4Dt\,}}\right)\\
\cdfsisot{1}{\fixpnew{2}{\ell}}{t}&=\frac{r_{r1}}{r_{r1}\!+\!d_{\fixpnew{2}{\ell}}^{\rxJ{1}}} \,\text{erfc} \left( \frac{d_{\fixpnew{2}{\ell}}^{\rxJ{1}}}{\sqrt{4Dt\,}}\right),
\end{split}
\end{align}
where $d_{\fixpnew{1}{\ell}}^{\rxJ{2}}$ is the distance from $\fixpnew{1}{\ell}$ to $\rxJ{2}$, and $d_{\fixpnew{2}{\ell}}^{\rxJ{1}}$ is the distance from $\fixpnew{2}{\ell}$ to $\rxJ{1}$.
By solving (\ref{cdf_limit}) simultaneously, we obtain
\begin{align}
\begin{split}
\lim_{t\to\infty}\cdft{2}{\txJ{1}}{t} &=\lim_{t\to\infty}\frac{\cdfsisot{2}{\txJ{1}}{t}-\cdfsisot{2}{\fixpnew{1}{\ell}}{t}\cdfsisot{1}{\txJ{1}}{t}}{1-\cdfsisot{2}{\fixpnew{1}{\ell}}{t}\cdfsisot{1}{\fixpnew{2}{\ell}}{t}} \\
\lim_{t\to\infty}\cdft{1}{\txJ{1}}{t} &=\lim_{t\to\infty}\frac{\cdfsisot{1}{\txJ{1}}{t}-\cdfsisot{1}{\fixpnew{2}{\ell}}{t}\cdfsisot{2}{\txJ{1}}{t}}{1-\cdfsisot{2}{\fixpnew{1}{\ell}}{t}\cdfsisot{1}{\fixpnew{2}{\ell}}{t}}.
\end{split}
\label{proofobj} 
\end{align}
We obtain $d_{\fixpnew{1}{\ell}}^{\rxJ{2}}$ and $d_{\fixpnew{2}{\ell}}^{\rxJ{1}}$ from the fact that~\eqref{eq_capture_prob} is equal to~\eqref{proofobj} (See Appendix A). 

Fianlly, we derive the approximation forms of $\cdft{1}{\txJ{1}}{t}$ and $\cdft{2}{\txJ{1}}{t}$ based on \eqref{cdf_1_app}, as follows (see Appendix B):
	\begin{align}
	\begin{split}		
	&\cdft{i}{\txJ{1}}{t}=\frac{-c_{i,1}}{\sqrt{\pi t}}\!\sum_{n=-1}^{-\infty}\bigg[\big\{(A^n\!\!-\!\!1)(\frac{d_{\txJ{1}}^{\rxJ{i}}}{\sqrt{D}}\!\!-\!\!u)+\frac{B}{(A\!\!-\!\!1)}(nA^{n+1} \\
	&\!\!-\!(n+1)A^n\!\!+\!\!1)\big\}e^{\frac{-u^2}{4t}} + (A^n\!\!-\!\!1)\sqrt{\pi t} \mathrm{erf} \left(\frac{u}{\sqrt{4t}}\right)\bigg]^{\frac{d_{\txJ{1}}^{\rxJ{i}}}{\sqrt{D}}+\!nB}_{\frac{d_{\txJ{1}}^{\rxJ{i}}}{\sqrt{D}}+\!(n\!+\!1)B} \\
	&+\frac{c_{i,2}}{\sqrt{\pi t}}\!\sum_{n=-1}^{-\infty}\bigg[\big\{(A^n\!\!-\!\!1)(b_i\!\!-\!\!u)+\frac{B}{(A\!\!-\!\!1)}(nA^{n\!+\!1} 
	\!\!\\
	&-\!(n\!+\!1)A^n\!\!+1)\big\}e^{\frac{-u^2}{4t}} + (A^n\!\!-\!\!1)\sqrt{\pi t} \mathrm{erf} \left(\frac{u}{\sqrt{4t}}\right)\bigg]^{b_i+\!nB}_{b_i+\!(n\!+\!1)B}, \\	
	\end{split} 
	\label{app_final}
	\end{align}
	where $A=\frac{(r_{r1}+d_{\fixpnew{2}{\ell}}^{\rxJ{1}})(r_{r2}+d_{\fixpnew{1}{\ell}}^{\rxJ{2}})}{r_{r1}r_{r2}}$, $B=-\left(d_{\fixpnew{1}{\ell}}^{\rxJ{2}}+d_{\fixpnew{2}{\ell}}^{\rxJ{1}}\right)/\sqrt{D}$,
	$b_1=\left(d_{\fixpnew{2}{\ell}}^{\rxJ{1}}+d_{\txJ{1}}^{\rxJ{2}}\right)/\sqrt{D}$, 
	$b_2=\left(d_{\fixpnew{1}{\ell}}^{\rxJ{2}}+d_{\txJ{1}}^{\rxJ{1}}\right)/\sqrt{D}$,
	$c_{1,1}=\frac{Ar_{r1}}{(A-1)(r_{r1}+d_{\txJ{1}}^{\rxJ{1}})}$, 
	$c_{1,2}=\frac{Ar_{r1}r_{r2}}{(A-1)(r_{r2}+d_{\txJ{1}}^{\rxJ{2}})(r_{r1}+d_{\fixpnew{2}{\ell}}^{\rxJ{1}})}$,
	$c_{2,1}=\frac{Ar_{r2}}{(A-1)(r_{r2}+d_{\txJ{1}}^{\rxJ{2}})}$, $c_{2,2}=\frac{Ar_{r1}r_{r2}}{(A-1)(r_{r1}+d_{\txJ{1}}^{\rxJ{1}})(r_{r2}+d_{\fixpnew{1}{\ell}}^{\rxJ{2}})}$, 
	and  $\big[H(u)\big]^a_b=H(a)-H(b)$.

 \begin{figure}[!t]
 	\begin{center}
 		\includegraphics[width=0.98\columnwidth,keepaspectratio]%
 		{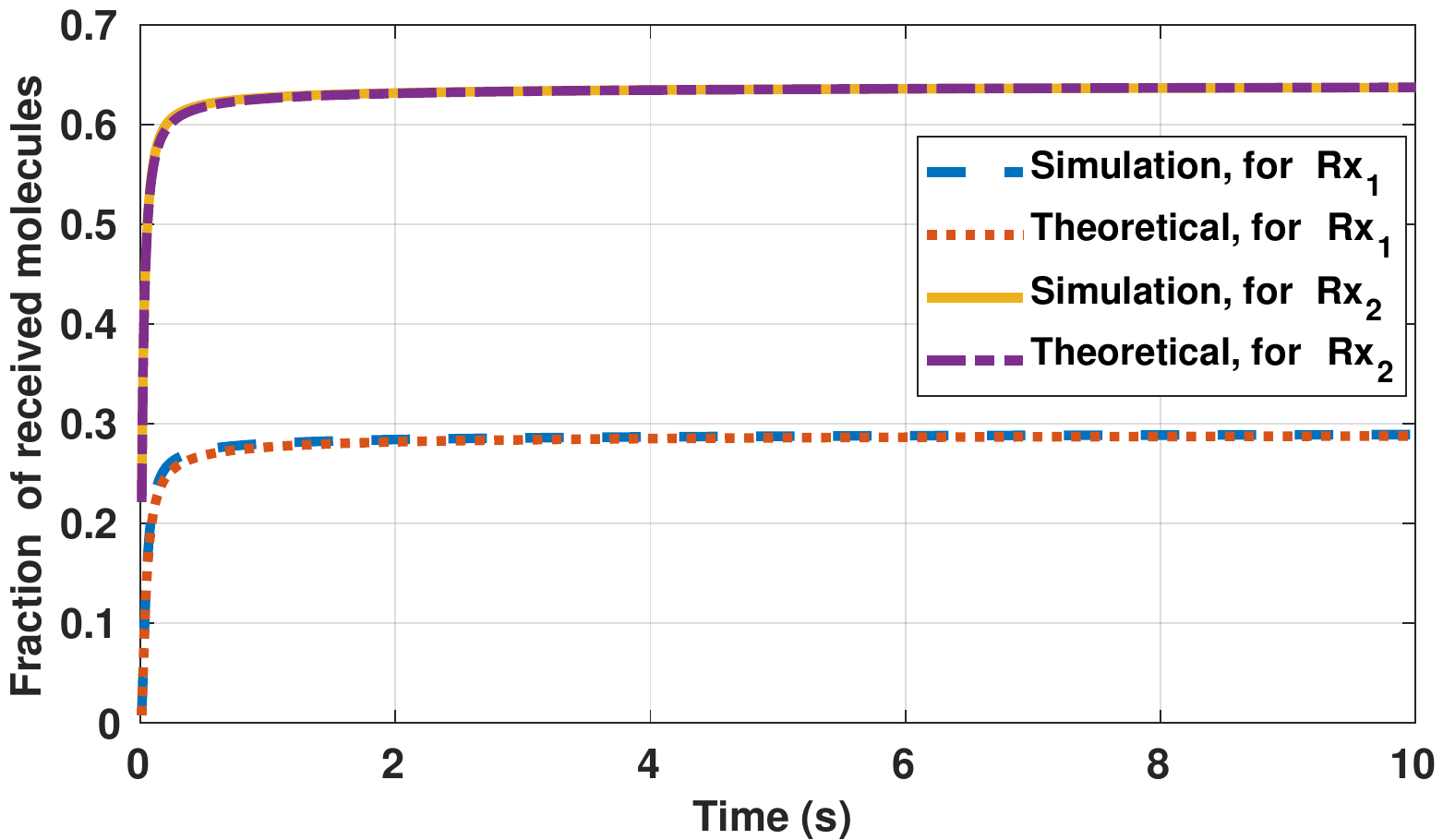}
 		\caption{A comparison of the theoretical fraction of received molecules with the simulation data ($d_1$=$d_2$=$\SI{1.5}{\micro\meter}$).} 
 		\label{fig_ch_fitting}
 	\end{center}
 \end{figure}
 \begin{table}
 	\caption{Simulation Parameters for Channel Model Verification}
 	\begin{center}
 		\begin{tabular}{lcc}
 			\hline\hline
 			System Parameter &Notation &Values \\
 			\hline
 			Distance ($\txJ{i}$ to $\rxJ{i}$) &$d_1$&\SI{1.5}{\micro\meter} \\
 			Distance ($\txJ{i}$ to $\rxJ{j}$)&$\dptRxJ{\txJ{1}}{2}$ &\SI{5}{\micro\meter}-$d_1$\\
 			Number of Molecules for Bit 1 &$\Ntx$ &50000 \\
 			Diffusion Coefficient &D &$\SI{100}{\micro\meter^2/\second}$ \\
 			Radius of Receiver &$r_{r1}$=$r_{r2}$ & \SI{5}{\micro\meter} \\
 			Simulation Time Step & $\Delta t$ & $10^{-5}\si{\second}$ \\
 			Simulation Duration & &\SI{0.1}{\second} \\
 			Simulation Replication & &10 \\
 			Number of the terms in \eqref{eq_capture_prob} and \eqref{app_final} &m, n & 100000\\
 			\hline \hline
 		\end{tabular}
 	\end{center}
 	\label{tab_simulation_parameter}
 \end{table}

\subsection{Channel Model Verification}

In this Section, we verify our channel model by particle-based simulation. We repeat the simulation for a number of times, and take the mean value with respect to the number of trials. The number of trials is presented in Tables II and III. In each simulated trial, 50000 molecules are released. The received molecules are distinguished  according to the transmitter that emits them. In the simulations, only $\txJ{1}$ releases molecules which is sufficient to verify the theoretical channel model.


To implement Brownian motion for the emitted molecules, our simulator records and updates the position of each molecule at each time step $\Delta t$. The position of the emitted molecules, $\mathbf{X}_p(t)$, changes by $\Delta \mathbf{X}_p$ after simulation time step $\Delta t$ as in (\ref{eq_simul})\cite{birkan2014simulation}. The simulation parameters used for verification of the channel model are given in Table~\ref{tab_simulation_parameter}.
\begin{align}
\begin{split}
\mathbf{X}_p(t+\Delta t)&=\mathbf{X}_p(t)+\Delta \mathbf{X}_p \\
          &=\mathbf{X}_p(t)+(\Delta x,\Delta y,\Delta z)  \\
\Delta x,\Delta y,\Delta z &\sim \mathscr{N}(0,2D\Delta t).
\end{split}
\label{eq_simul}
\end{align}

Through extensive simulations, we obtain the fraction of received molecules for each receiver (i.e., $\rxJ{1}$ and $\rxJ{2}$) at each time step during the simulation time, i.e., $\cdft{2,\text{sim}}{\txJ{1}}{t}$ and $\cdft{1,\text{sim}}{\txJ{1}}{t}$, respectively. Then, we use (\ref{app_final}) to calculate $\cdft{2}{\txJ{1}}{t}$ and $\cdft{1}{\txJ{1}}{t}$. We compare the analytical results with the simulation results, as shown in Fig.~\ref{fig_ch_fitting}. It is worth noting that the asymptotic capture probabilities (i.e., $\lim_{t\to\infty}\cdft{1}{\txJ{1}}{t}$ and $\lim_{t\to\infty}\cdft{2}{\txJ{1}}{t}$ in~\eqref{eq_capture_prob}) for Fig.~\ref{fig_ch_fitting} are $0.6414$ and $0.2932$, respectively.
By the long-term simulation (i.e., $t=1000s$), we observe that the simulation results and~\eqref{app_final} converge to those values.
We depict the number of received molecules in Fig.~\ref{fig_ch_fitting_f}, which corresponds to the impulse response of the channel. Theoretical values in Fig.~\ref{fig_ch_fitting_f} at $t=k\Delta t$ are obtained by calculating 
\begin{equation}
\Ntx\left(\cdft{i}{\txJ{1}}{(k+1) \Delta t}-\cdft{i}{\txJ{1}}{k\Delta t}\right).
\end{equation}
Using the derived channel model, we substitute $\pksys{ij}{k}$ and $\phipksys{ij}{k}$ as follows:
\begin{align}
\begin{split}
&\pksys{ij}{k} = \cdft{j}{\txJ{i}}{(k+1)t_s}-\cdft{j}{\txJ{i}}{kt_s},\\
&\phipksys{ij}{k} = \cdft{j}{\txJ{i}}{(k+1)t_s}-\cdft{j}{\txJ{i}}{kt_s+T_c}.
\end{split}
\end{align}
In Fig.~\ref{fig_ch_taps}, we depict the comparison of the channel coefficients $\pksys{ij}{k}$.

\begin{figure}[!t]
	\begin{center}
		\includegraphics[width=0.98\columnwidth,keepaspectratio]%
		{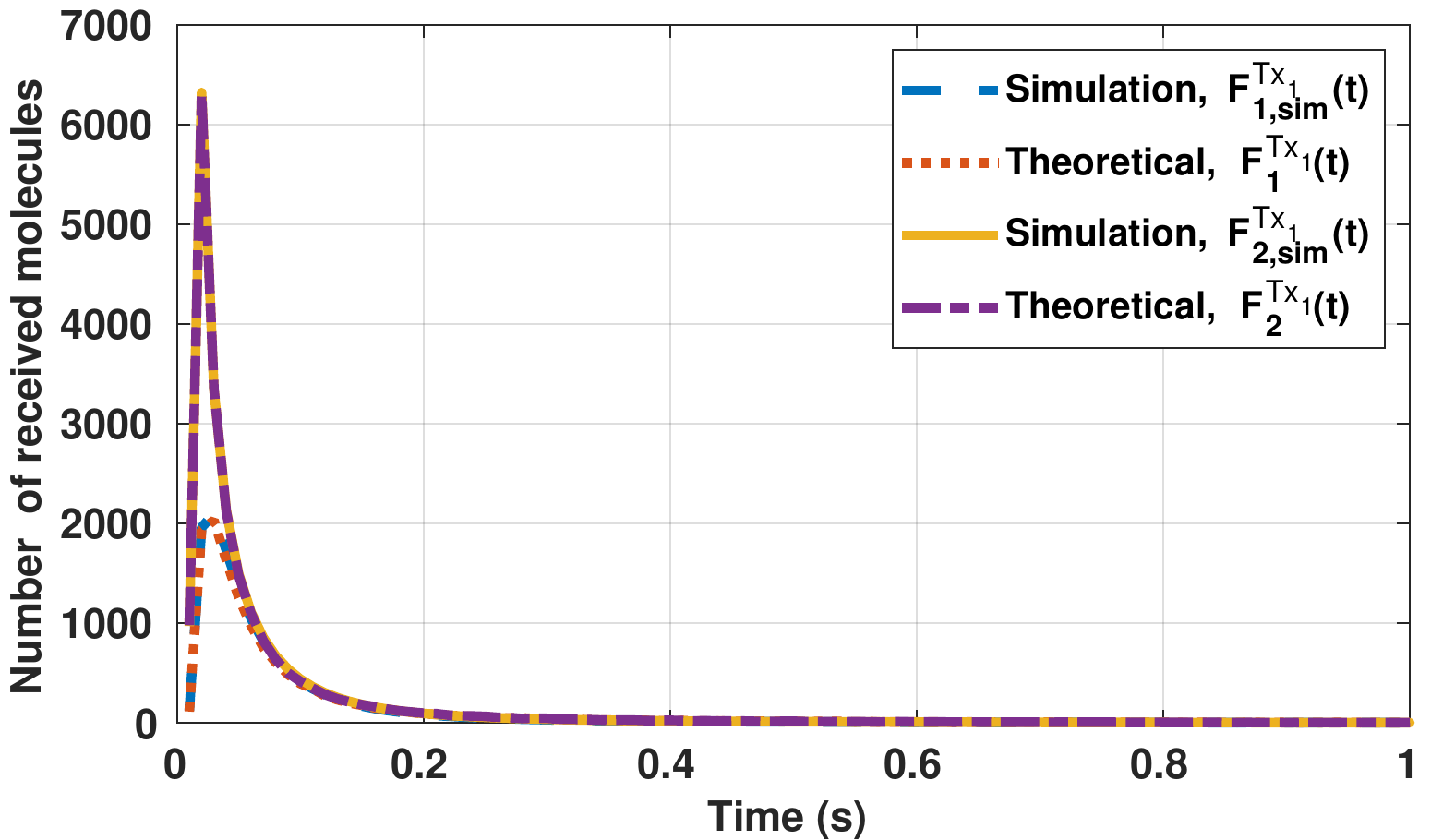}
		\caption{A comparison of the theoretical number of received molecules with the simulation data ($d_1$=$d_2$=$\SI{1.5}{\micro\meter}$).} 
		\label{fig_ch_fitting_f}
	\end{center}
\end{figure}
\begin{figure}[!t]
	\begin{center}
		\includegraphics[width=0.98\columnwidth,keepaspectratio]{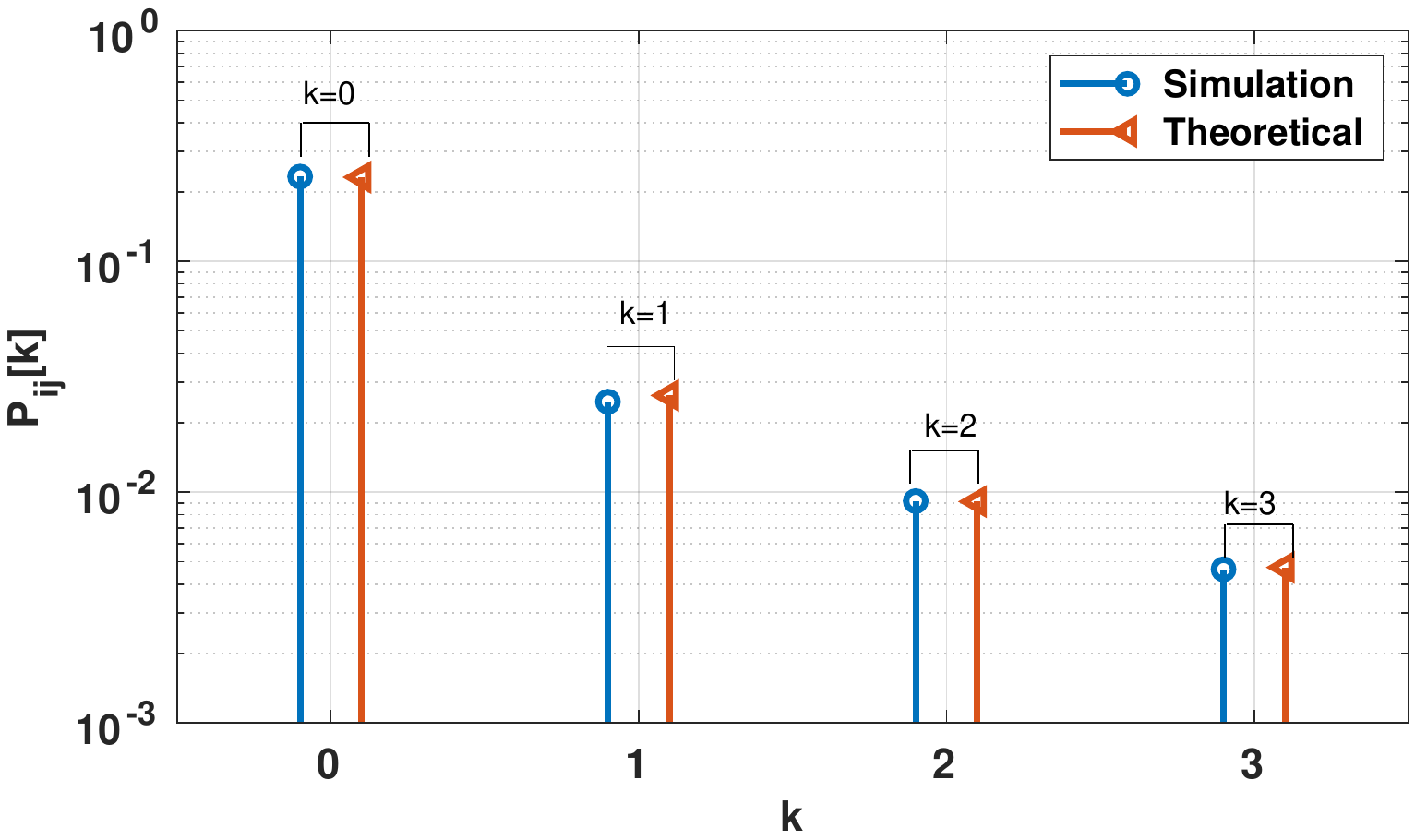}%
		
		\caption{A comparison of the theoretical channel coefficients with the simulation data ($t_s=\SI{0.150}{\second}$, $d_1$=$d_2$=$\SI{1.5}{\micro\meter}$,\ $r_{r1}$=$r_{r2}$=$\SI{5}{\micro\meter}$).}
		\label{fig_ch_taps}
	\end{center}
\end{figure}

\subsection{BER Formula for Two-Way Molecular Communication}

In this section, we formulate the BER in terms of the Q-function (i.e., the tail probability of the standard normal distribution). For the receiver $\rxJ{j}$, an error occurs when the result of decoding is different from the bit transmitted from the $\txJ{i}$. If $\txJ{i}$ encodes bit-1, an error occurs when $\yrxJN{j}{n}$ is less than the detection threshold $\tau_d$. If $\txJ{i}$ encodes bit-0, then an error occurs when $\yrxJN{j}{n}$ is greater than the detection threshold $\tau_d$. 

\begin{figure*}[!t]
	\begin{center}	
		\subfigure[Half-duplex system]	{\includegraphics[width=0.94\columnwidth,keepaspectratio]%
			{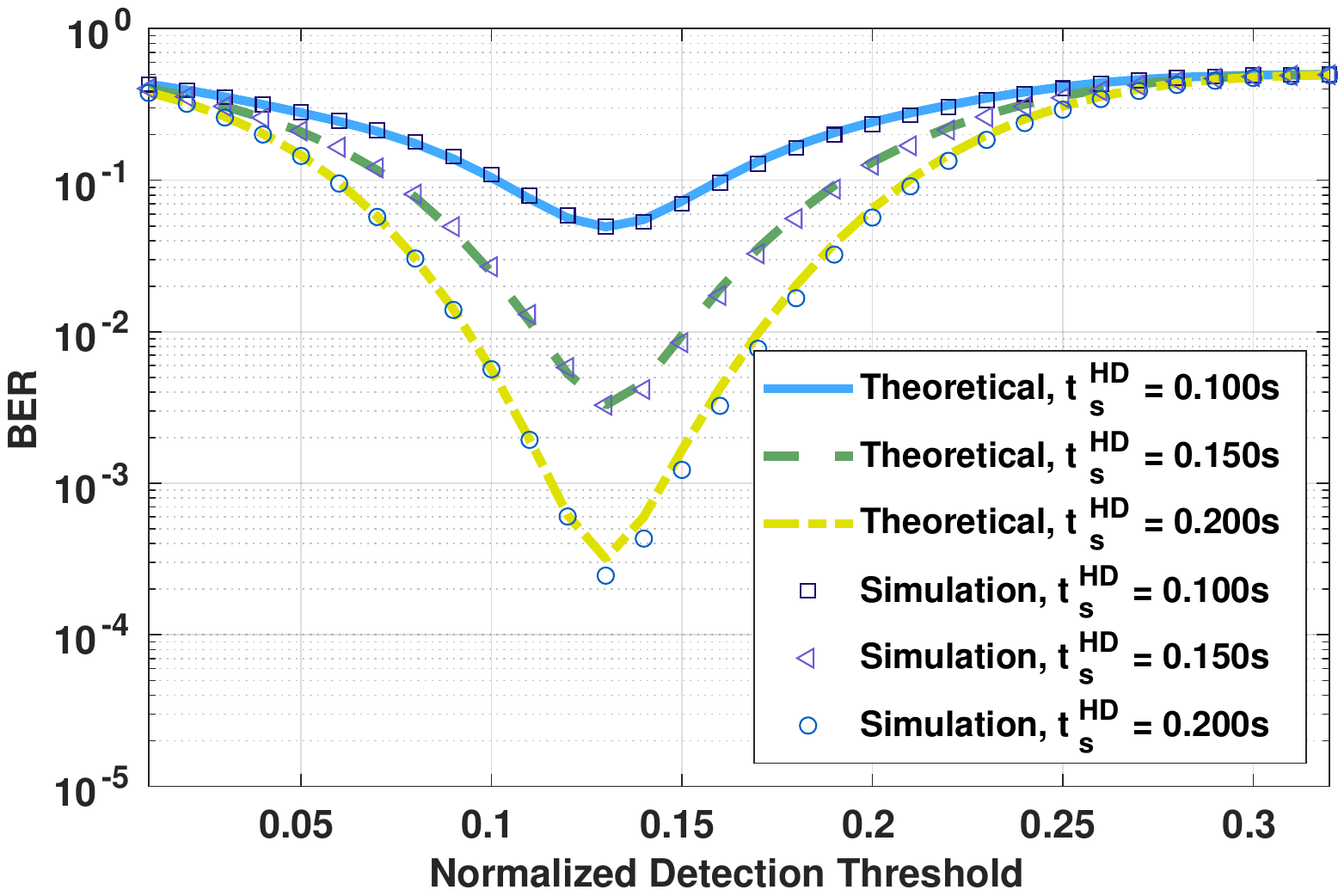}
			\label{fig_ber_thr_hd}}{\phantom{123}}
		\subfigure[Full-duplex system]
		{\includegraphics[width=0.94\columnwidth,keepaspectratio]%
			{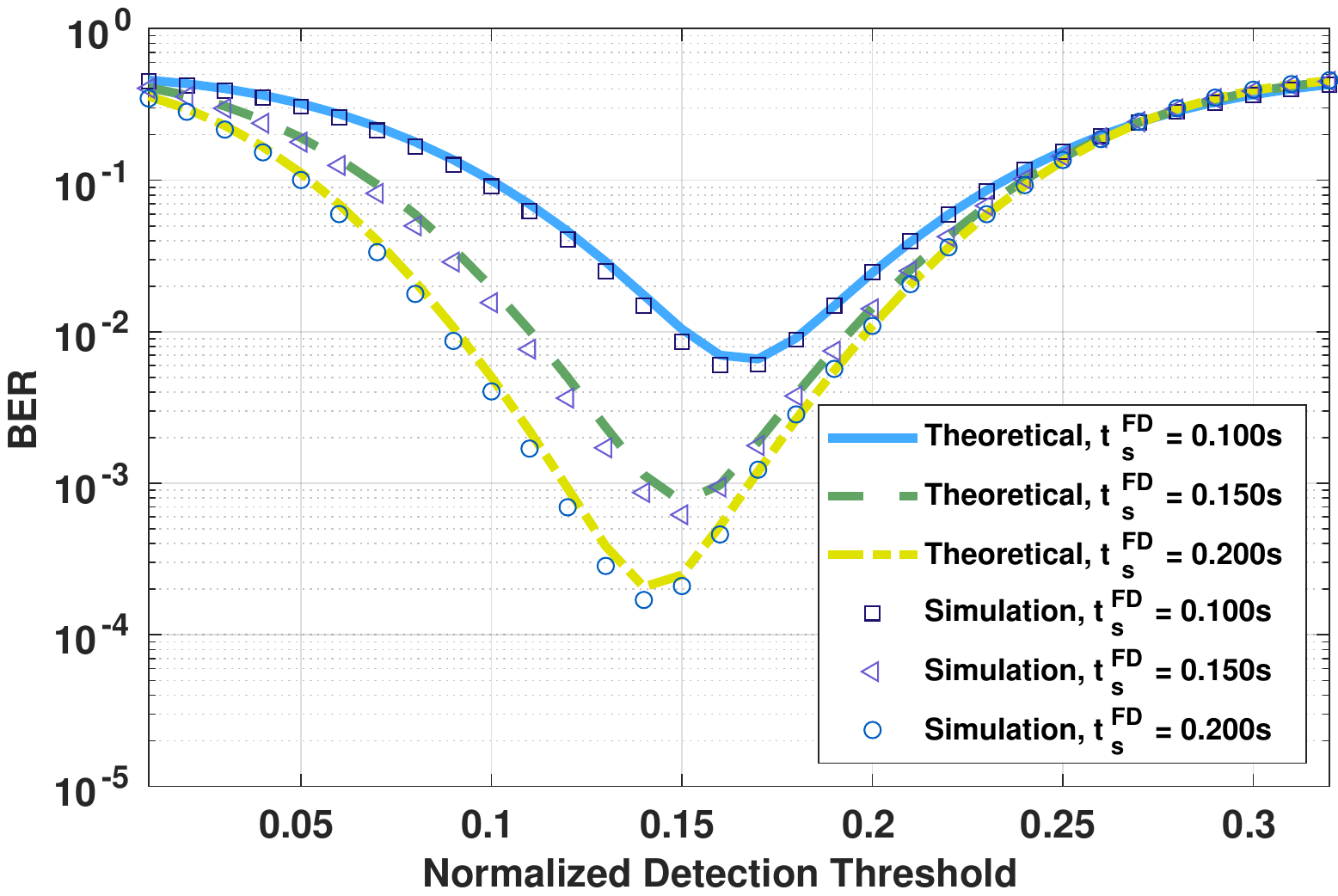}
			\label{fig_ber_thr}}
		\caption{(a)BER comparison of the simulation data and the theoretical analysis of the half-duplex system using BCSK (b) BER comparison of the simulation data and the theoretical analysis of the full-duplex system with D-SIC ($\Ntx$=500, $d_1$=$d_2$=$\SI{1.5}{\micro\meter}$,\ $r_{r1}$=$r_{r2}$=$\SI{5}{\micro\meter}$).}
	\end{center}
\end{figure*}
\begin{table}[t]
	\caption{Simulation Parameters for BER and Throughput Comparison}
	\label{table_3}
	\begin{center}
		\begin{tabular}{ccc}
			\hline\hline
			System Parameter &Notation &Values \\
			\hline
			Distance ($\txJ{i}$ to $\rxJ{i}$) &$d_1$=$d_2$ &\SI{1.5}{\micro\meter}\\
			Distance ($\txJ{i}$ to $\rxJ{j}$)&$\dptRxJ{\txJ{2}}{1}$=$\dptRxJ{\txJ{1}}{2}$ &\SI{5}{\micro\meter}-$d_1$\\
			Number of Molecules for Bit 1 &$\Ntx$ &300, 400, 500 \\
			Diffusion Coefficient &D &$\SI{100}{\micro\meter^2/\second}$ \\
			Radius of Receiver &$r_r$ & $\SI{5}{\micro\meter}$ \\
			Simulation Time Step & $\Delta t$ & $10^{-5}\si{\second}$ \\
			Molecular Noise Variance &$\sigma_{\text{noise}}^2$ & $100$\\
			Considered ISI Period & & 0.6s\\ 
			Number of the terms in \eqref{eq_capture_prob} and \eqref{app_final} &m, n & 100000\\
			Replication & &10 \\
			\hline \hline
		\end{tabular}
	\end{center}
\end{table}
Considering the transmitted bit sequences $x_i$ and $x_j$, we obtain the error probabilities at the $n\text{th}$ symbol slot as (\ref{eq_ber}) where $x_i[1\!:\! n \!\sminus\! 1]$ denotes the bits transmitted previously from $\txJ{i}$:
\begin{align*}
\pe{j,BCSK} &=\pec{x_i[n]=1}{j,BCSK} +\pec{x_i[n]=0}{j,BCSK}  \\
       =& \!\!\sum_{ x_i[1:n\!\sminus\!1]}\sum_{ x_j}
         {P}_{i,1} \,P(\yrxJNsys{j}{n} \!\leq\!\tau_d \boldsymbol{|} x_i[n]\!=\!1,x_i,x_j) \\
       +& \!\sum_{ x_i[1:n\!\sminus\!1]}\sum_{ x_j}
         {P}_{i,0} \,{P}(\yrxJNsys{j}{n} \!>\!\tau_d \boldsymbol{|} x_i[n]\!=\!0,x_i,x_j) \\
       =&  \!\sum_{ x_i[1:n\!\sminus\!1]}\sum_{ x_j}
        {P}_{i,1} \,Q\left( \frac{\mu_{n,\text{total}|x_i,x_j}\!\sminus\!\tau_d}
                {\sigma^2_{n,\text{total}|x_i,x_j}}\right) \\
       +&  \!\sum_{ x_i[1:n\!\sminus\!1]}\sum_{ x_j}
       {P}_{i,0} \,Q\left( \frac{\tau_d \!\sminus\! \mu_{n,\text{total}|x_i,x_j}}
                {\sigma^2_{n,\text{total}|x_i,x_j}}\right),  \numberthis
\label{eq_ber}
\end{align*}
where $\mu_{n,\text{total}}$ and $\sigma^2_{n,\text{total}}$ are defined in~\eqref{yrxjn}, $Q(\cdot)$ is the Q-function, and 
\begin{align}
\begin{split}
{P}_{i,1} &= {P}(x_i[n]\!=\!1) \, {P}(x_i[1\!:\!n\!\sminus\!1]) \, {P}(x_j[1:n]) \\
{P}_{i,0} &= {P}(x_i[n]\!=\!0) \, {P}(x_i[1\!:\!n\!\sminus\!1]) \, {P}(x_j[1:n]).
\end{split}
\label{eq_summation}
\end{align}
Note that $\sum_{ x_i[1:n\!\sminus\!1]}$ and $\sum_{ x_j}$ in~\eqref{eq_ber} indicate the summation over all possible bit sequences $x_i[1\!\!:\!\! n \!\sminus\!1]$ and $x_j[1\!\!:\!\!n]$, respectively. The BER expression for the full-duplex system with the proposed SIC techniques is obtained by substituting $\yrxJN{j}{n}$, $\mu_{n,\text{total}|x_i,x_j}$ and $\sigma^2_{n,\text{total}|x_i,x_j}$ in~\eqref{eq_summation} with $\yrxJNphi{j}{n}$, $\mu^{\varphi}_{n,\text{total}|x_i,x_j}$, and $(\sigma^{\varphi}_{n,\text{total}})^2$, respectively.

\section{Particle-Based Simulation Results}
\label{sec_numerical_results}
In this section, we compare the proposed half-duplex and full-duplex systems in terms of BER and throughput. Throughput of the systems are evaluated as follows:
\begin{equation}
Throughput=\frac{M(1-{P}_e)}{t_s},
\label{eq_throughput}
\end{equation}
where $M$ is the number of bits transmitted in one symbol, ${P}_e$ is the BER of the system, and $t_s$ denotes the symbol duration. The system parameters for the rest of our work are summarized in Table~\ref{table_3}. For convenience, we denote $t_s$ of the half-duplex and full-duplex systems as $t^{\text{HD}}_s$ and $t^{\text{FD}}_s$, respectively. In Table~\ref{table_bcsk_bcsk} and Table~\ref{table_qcsk_bcsk}, throughput ratio is calculated as a ratio of the throughput of full-duplex system and the half-duplex system.


In the half-duplex system, each receiver operates only when the paired transmitter releases the molecular signal. Hence, the operating time of the receiver (i.e., detection period) is half of the symbol duration. In the full-duplex system, the detection period of each receiver is equal to the symbol duration. Roughly, we can expect faster but less accurate communications in the full-duplex system if we use the same modulation technique and detection period for both systems. By the theoretical and simulation BER analysis, we confirm that the proposed SIC techniques are necessary in the full-duplex system. Therefore we analyze the BER improvement in the full-duplex system with SIC to find, numerically, the optimal values for the normalized detection threshold and the discarding time of the SIC to minimize BER. Through the numerical parameter optimization, the throughput of the full-duplex system with SIC is compared to the half-duplex system. For a fair comparison, we evaluate the throughput in the following three cases, considering that the throughput is a function of M, $t_s$, and $P_e$:
\begin{enumerate}
\item Half-duplex system (BCSK) vs. full-duplex system (BCSK) with SIC, where $t^{\text{FD}}_s\!=\!t^{\text{HD}}_s/2$.
\item Half-duplex system (BCSK) vs. full-duplex system (BCSK) with SIC. We set the same $t_s$ for both systems (i.e., $t^{\text{FD}}_s\!=\!t^{\text{HD}}_s$).
\item Half-duplex system (BCSK) vs. full-duplex system (BCSK) with SIC. We empirically adjust $t^{\text{FD}}_s$ to make the BER of both systems close.
\item Half-duplex system using quadrature concentration shift keying (QCSK) vs full-duplex system (BCSK) with SIC, where $t^{\text{FD}}_s\!=\!t^{\text{HD}}_s/2$.
\end{enumerate}


\subsection{BER Analysis} 
Fig.~\ref{fig_ber_thr_hd} depicts the simulation and theoretical BERs of the half-duplex system using BCSK. The \emph{x}-axis is the normalized threshold ($\tau_m$), which is $\tau_d/\Ntx$. First of all, the simulation and theoretical values match each other well. 
Since the half-duplex system is not susceptible to SI, we do not need to apply the proposed SIC techniques to this system. On the other hand, we observe from Fig.~\ref{fig_sic} that the BER of the full-duplex system is nearly 0.3 if we do not apply the SIC techniques. In Fig.~\ref{fig_ber_thr_hd}, we can see the optimal normalized threshold $\tau_m^*$ for different $t^{\text{HD}}_s$ and we observe that it is slightly changing according to the value of $t^{\text{HD}}_s$. We also observe that the BER gain is relatively higher for changing $t^{\text{HD}}_s$ from \SI{0.100}{\second} to \SI{0.150}{\second} compared to from \SI{0.150}{\second} to \SI{0.200}{\second} due to the relative ISI difference.

\begin{figure}[!t]
	\begin{center}
		\includegraphics[width=0.94\columnwidth,keepaspectratio]%
		{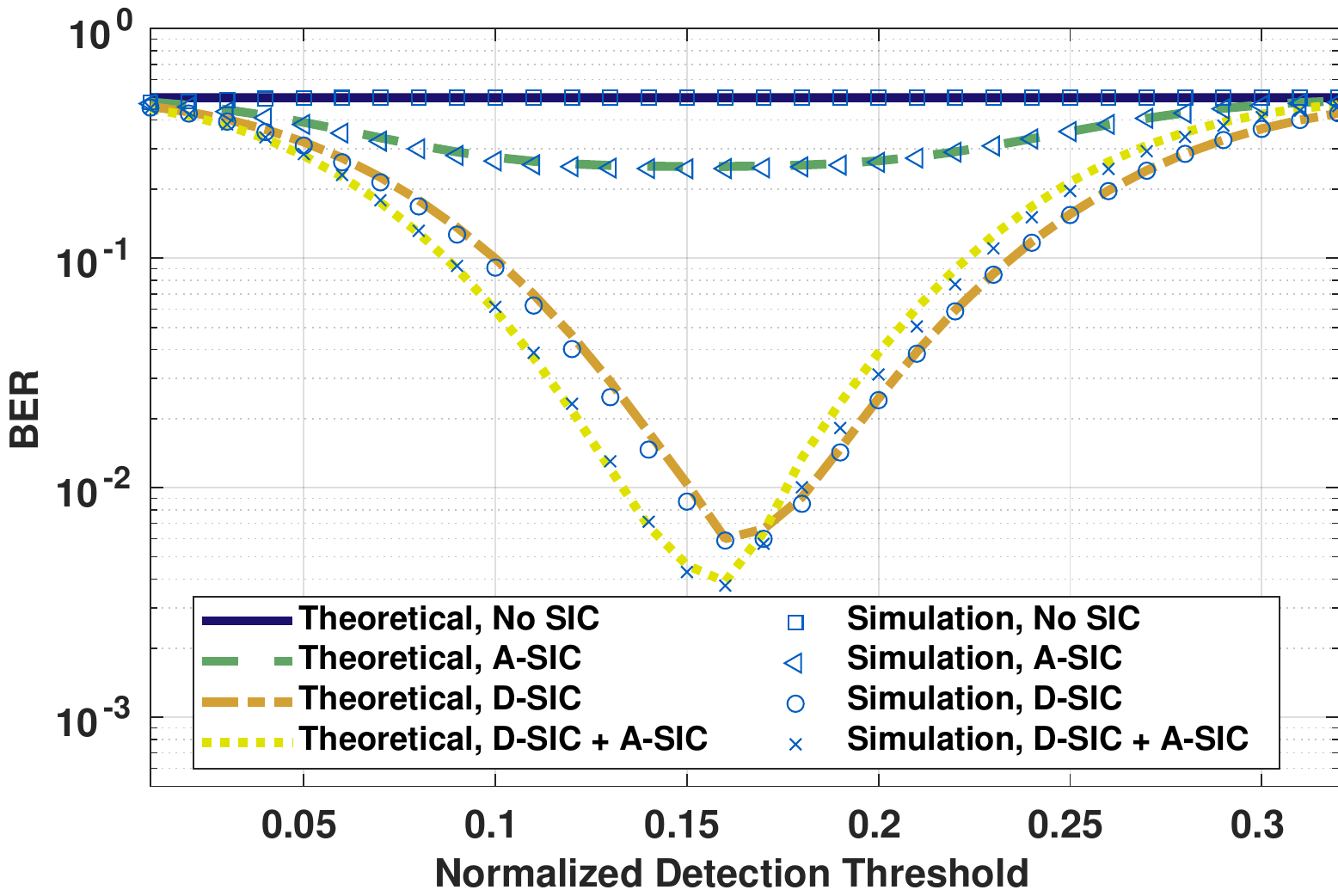}
		\caption{BER of the full-duplex system with the different SIC techniques ($\Ntx$=500, $t_s^{\rm{FD}}$= 0.1s, $r_{r1}$=$r_{r2}$=$\SI{5}{\micro\meter}$). For the no-SIC case, we set $d_1$=$d_2$=$\SI{0}{\micro\meter}$. For other cases, we set $d_1$=$d_2$=$\SI{1.5}{\micro\meter}$.}
		\label{fig_sic}
	\end{center}
\end{figure}
Fig.~\ref{fig_ber_thr} shows the simulation and theoretical BER of the full-duplex system with D-SIC. We can see that there is an optimal normalized detection threshold $\tau_m^*$ for different $t^{\text{FD}}_s$ values.
Similar to the case of the half-duplex system, there is some similarity between the $\tau_m^*$ changes according to the $t^{\text{FD}}_s$ value and also the tendency of the BER gain with respect to
\begin{figure}[!t]
	\begin{center}
		\includegraphics[width=0.82\columnwidth,keepaspectratio]%
		{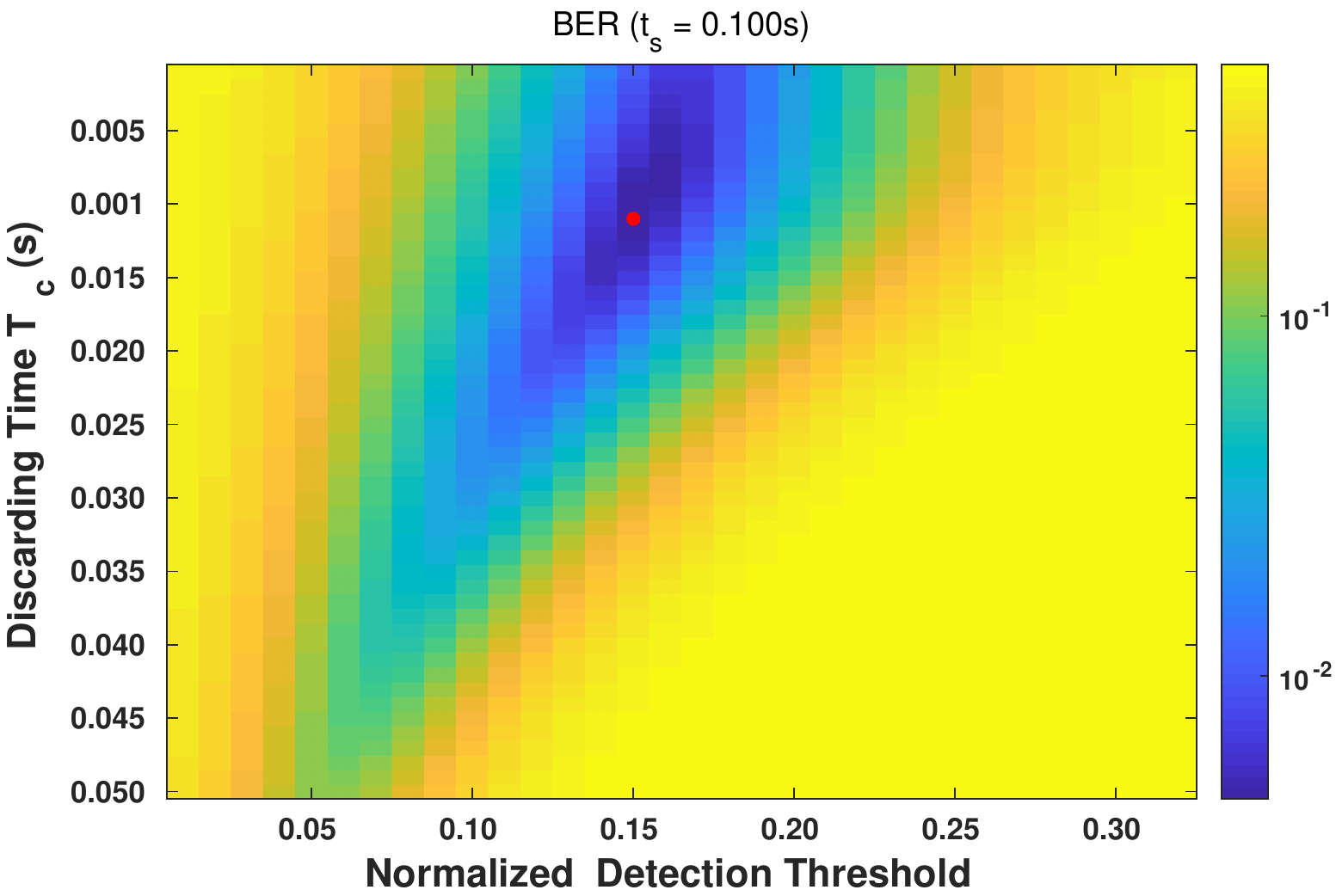}      
		\caption{Theoretical BER heatmap of the full-duplex system with A-SIC and D-SIC. The red mark indicates the optimal parameters for the normalized detection threshold and the discarding time.}
		\label{fig_heatmap}  
	\end{center}
\end{figure}
the $t^{\text{FD}}_s$ difference (see Fig.~\ref{fig_ber_thr_hd}).


Fig.~\ref{fig_sic} depicts the simulation and theoretical BERs of the full-duplex system while applying the different SIC techniques. The \emph{x}-axis is the normalized detection threshold $\tau_d/\Ntx$. 
The performance of D-SIC is superior to the performance of A-SIC, because D-SIC deals with a wider time interval (i.e., when the discarding time is $T_c$, A-SIC ignores the molecules received in $[0,T_c]$, and D-SIC subtracts the expected number of the SI molecules received in $[T_c,t_s]$). For the no-SIC case, we assume that $\txJ{i}$ is located on the surface of $\rxJ{i}$ (i.e., $d_1$=$d_2$=$\SI{0}{\micro\meter}$), and we observe that the BER is nearly 0.5. The results indicate that the separation of transmitter and receiver is necessary in the two-way MCvD system.
We observe that the BER of the full-duplex system with A-SIC is nearly 0.3. On the other hand, the BER of the full-duplex system with D-SIC becomes comparable with the half-duplex system in Fig.\ref{fig_ber_thr_hd}. 
Moreover, we observe that when we apply both SIC techniques, the BER is slightly improved compared to the full-duplex system with only D-SIC. 
As we stated in~\eqref{yrxjn} and~\eqref{phiyrxjn}, the number of received SI molecules from the current symbol can be modeled as a Gaussian random variable that has a variance $\Ntx\phipksys{ii}{0}(1-\phipksys{ii}{0})x_i[n]$. 
Since the variance decreases as $\phipksys{ii}{0}$ decreases (when $\phipksys{ii}{0}\leq 0.5$), applying A-SIC can reduce the uncertainty of predicting the number of received SI molecules from the current symbol.
For the full-duplex systems, since the BER of the A-SIC-only system and the no-SIC system are not reasonable, we consider those full-duplex systems with only D-SIC or D-SIC and A-SIC. 



As derived in Section~\ref{sec_ber_formulations}, the BER is a function of symbol duration ($t_s$), detection threshold ($\tau_d$), the number of molecules for encoding bit-1 ($\Ntx$), and the discarding time of molecular received signal for A-SIC ($T_c$). Since $\Ntx$ and $t_s$ are system parameters, we consider only $\tau_d$ and $T_c$ as variables to optimize. While we will show that optimal values for these parameters exist, we cannot derive them in closed-form and hence must resort to evaluating them numerically. To improve BER, we first need to see the structure of the BER of the full-duplex system with SIC in terms of $\tau_d$ and $T_c$ by the following analysis. 

\begin{table*}
	\caption{Throughput ratio (FD/HD), $t_s$ and BER values of the half-duplex system using BCSK and the full-duplex system using BCSK.}
	\label{table_bcsk_bcsk}
	\begin{center}
		\begin{tabular}{cccccccccccccc}
			\hline\hline
			\multicolumn{2}{c|}{}&\multicolumn{4}{c|}{$\Ntx=300$} &\multicolumn{4}{c|}{$\Ntx=400$} &\multicolumn{4}{c}{$\Ntx=500$} \\ 
			\hline
			
			\multicolumn{1}{c|}{}&\multicolumn{1}{c|}{$t^{\text{HD}}_s\!=\!2t^{\text{FD}}_s$} &$\SI{0.100}{\second}$&$\SI{0.200}{\second}$&$\SI{0.300}{\second}$&\multicolumn{1}{c|}{$\SI{0.400}{\second}$}&$\SI{0.100}{\second}$&$\SI{0.200}{\second}$&$\SI{0.300}{\second}$&\multicolumn{1}{c|}{$\SI{0.400}{\second}$}&$\SI{0.100}{\second}$&$\SI{0.200}{\second}$&$\SI{0.300}{\second}$&\multicolumn{1}{c}{$\SI{0.400}{\second}$}\\
			
			\multicolumn{1}{c|}{$\text{Case 1}$}&\multicolumn{1}{c|}{BER (FD)} &0.1155&0.0229&0.0083&\multicolumn{1}{c|}{0.0045}&0.0887&0.0097&0.0022&\multicolumn{1}{c|}{0.0009}&0.0721&0.0045&$6.7\!\!\times\!\!10^{-\!4}$&$1.9\!\!\times\!\!10^{-\!4}$\\ 
			
			\multicolumn{1}{c|}{}&\multicolumn{1}{c|}{BER (HD)} &0.0931&0.0086&0.0022&\multicolumn{1}{c|}{0.0010}&0.0660&0.0017&0.0002&\multicolumn{1}{c|}{$5\!\!\times\!\!10^{-\!5}$}&0.0492&$3.1\!\!\times\!\!10^{-\!4}$&$1.5\!\!\times\!\!10^{-\!5}$&$1.5\!\!\times\!\!10^{-\!6}$ \\
			
			\multicolumn{1}{c|}{}&\multicolumn{1}{c|}{THP Ratio} &1.9506&1.9711&1.9876&\multicolumn{1}{c|}{1.9929}&1.9515&1.9840&1.9959&\multicolumn{1}{c|}{1.9983}&1.9517&1.9916&1.9987&1.9996 \\
			
			\hline
			
			\multicolumn{1}{c|}{$$}&\multicolumn{1}{c|}{$t^{\text{FD}}_s$=$t^{\text{HD}}_s$} &$\SI{0.100}{\second}$&$\SI{0.200}{\second}$&$\SI{0.300}{\second}$&\multicolumn{1}{c|}{$\SI{0.400}{\second}$}&$\SI{0.100}{\second}$&$\SI{0.200}{\second}$&$\SI{0.300}{\second}$&\multicolumn{1}{c|}{$\SI{0.400}{\second}$}&$\SI{0.100}{\second}$&$\SI{0.200}{\second}$&$\SI{0.300}{\second}$&\multicolumn{1}{c}{$\SI{0.400}{\second}$}\\
			
			\multicolumn{1}{c|}{$\text{Case 2}$}&\multicolumn{1}{c|}{BER (FD)} &0.0229&0.0045&0.0023&\multicolumn{1}{c|}{0.0014}&0.0097&0.0009&0.0003&\multicolumn{1}{c|}{0.0001}&0.0045&$1.9\!\!\times\!\!10^{-\!4}$&$4.7\!\!\times\!\!10^{-\!5}$&$1.1\!\!\times\!\!10^{-\!5}$\\ 
			
			\multicolumn{1}{c|}{}&\multicolumn{1}{c|}{BER (HD)} &0.0931&0.0086&0.0022&\multicolumn{1}{c|}{0.0010}&0.0660&0.0017&0.0002&\multicolumn{1}{c|}{$5\!\!\times\!\!10^{-\!5}$}&0.0492&$3.1\!\!\times\!\!10^{-\!4}$&$1.5\!\!\times\!\!10^{-\!5}$&$1.5\!\!\times\!\!10^{-\!6}$ \\
			
			\multicolumn{1}{c|}{}&\multicolumn{1}{c|}{THP Ratio} &1.0774&1.0041&0.9999&\multicolumn{1}{c|}{0.9997}&1.0603&1.0008&0.9999&\multicolumn{1}{c|}{0.9999}&1.0469&1.0001&1.0000&1.0000 \\
			
			\hline

			\multicolumn{1}{c|}{}&\multicolumn{1}{c|}{$t^{\text{HD}}_s$} &$\SI{0.100}{\second}$&$\SI{0.200}{\second}$&$\SI{0.300}{\second}$&\multicolumn{1}{c|}{$\SI{0.400}{\second}$}&$\SI{0.100}{\second}$&$\SI{0.200}{\second}$&$\SI{0.300}{\second}$&\multicolumn{1}{c|}{$\SI{0.400}{\second}$}&$\SI{0.100}{\second}$&$\SI{0.200}{\second}$&$\SI{0.300}{\second}$&\multicolumn{1}{c}{$\SI{0.400}{\second}$}\\
			
			\multicolumn{1}{c|}{$\text{Case 3}$}&\multicolumn{1}{c|}{$t^{\text{FD}}_s$} &$\SI{0.056}{\second}$&$\SI{0.147}{\second}$&$\SI{0.320}{\second}$&\multicolumn{1}{c|}{$\SI{0.630}{\second}$}&$\SI{0.056}{\second}$&$\SI{0.165}{\second}$&$\SI{0.390}{\second}$&\multicolumn{1}{c|}{$\SI{2.785}{\second}$}&$\SI{0.056}{\second}$&$\SI{0.178}{\second}$&$\SI{0.398}{\second}$&N/A\\ 
			
			\multicolumn{1}{c|}{$$}&\multicolumn{1}{c|}{THP Ratio} &1.7747&1.3603&0.9375&\multicolumn{1}{c|}{0.6349}&1.7853&1.2122&0.7692&\multicolumn{1}{c|}{0.1436}&1.7820&1.1236&0.7538&N/A \\
			\hline\hline
		\end{tabular}
	\end{center}
\end{table*}

In Fig.~\ref{fig_heatmap}, we depict a heatmap of the theoretical BER with respect to $T_c$ and $\tau_m$ for the full-duplex system with D-SIC and A-SIC. We observe that $\tau_m$ minimizing the BER with given $T_c$ decreases as $T_c$ increases. It is because that the molecules from the paired transmitter (i.e., desired signal) are also discarded by A-SIC.
For each $t^{\text{FD}}_s$, we find global optimal normalized detection threshold and the discarding time in order to minimize the BER and denote them as red mark on the heatmap. We utilized these optimal values in the SIC algorithm to compare the throughput of the half-duplex system and the full-duplex system with SIC.

As was mentioned before, for comparison, we consider the following three cases: i) set $t^{\text{FD}}_s\!=\!t^{\text{HD}}_s/2$ (to observe a trade off between the throughput and BER); ii) set the same $t_s$ for both systems; iii) fix $t^{\text{HD}}_s$ and adjust $t^{\text{FD}}_s$ empirically to make the BER of both systems close. For the case iii), adjusted $t^{\text{FD}}_s$ values are in Table~\ref{table_bcsk_bcsk}.


\begin{figure}[!t]
	\begin{center}
		\includegraphics[width=0.82\columnwidth,keepaspectratio]%
		{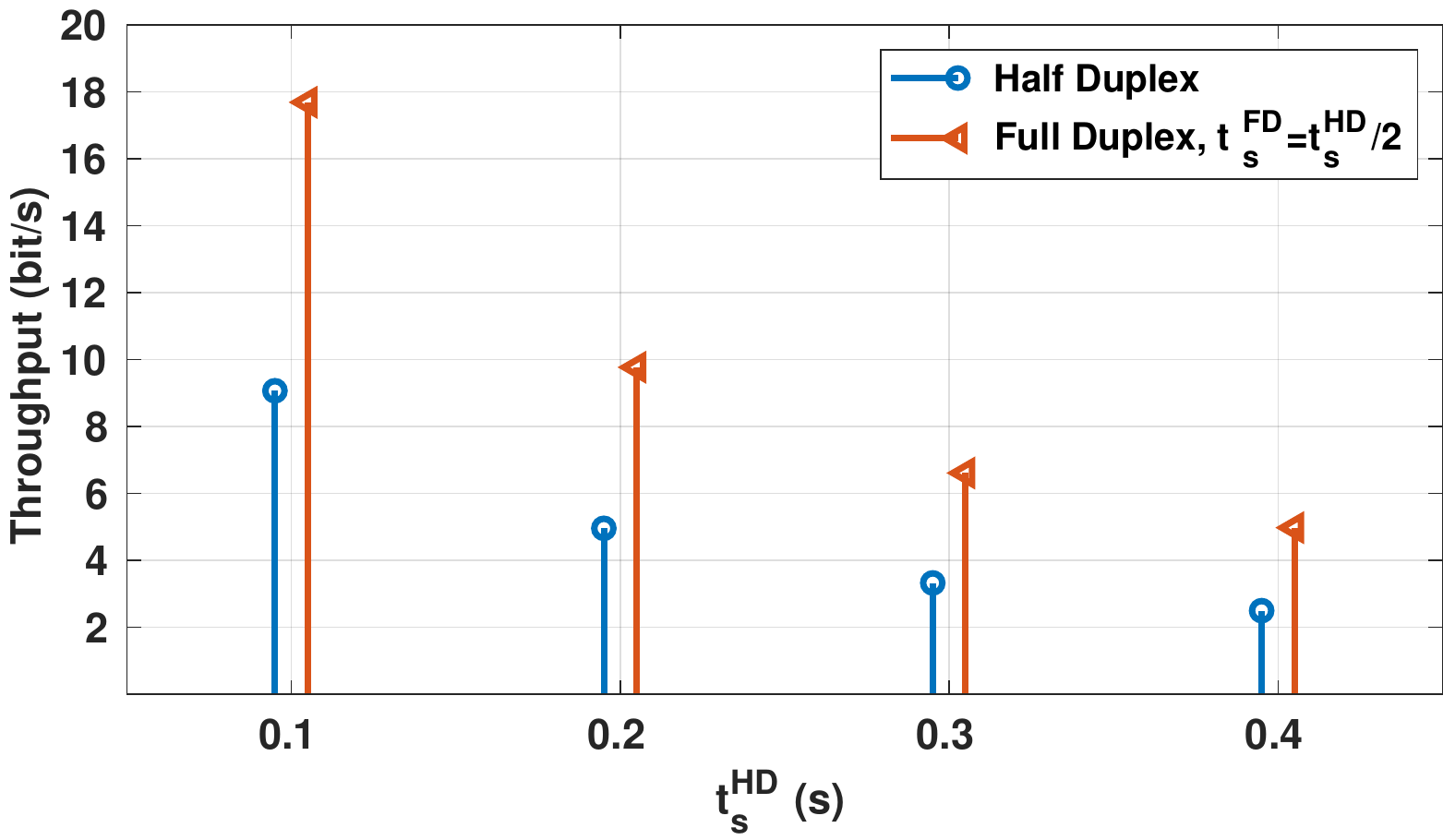}
		\caption{Theoretical BER of the half-duplex system (BCSK) and the full-duplex system with optimized D-SIC and A-SIC (BCSK), where $t^{\text{FD}}_s\!=\!t^{\text{HD}}_s/2$.}
		\label{fig_bcsk_bcsk_ts_half}
	\end{center}
\end{figure}
\begin{figure}[!t]
	\begin{center}
		\includegraphics[width=0.81\columnwidth,keepaspectratio]%
		{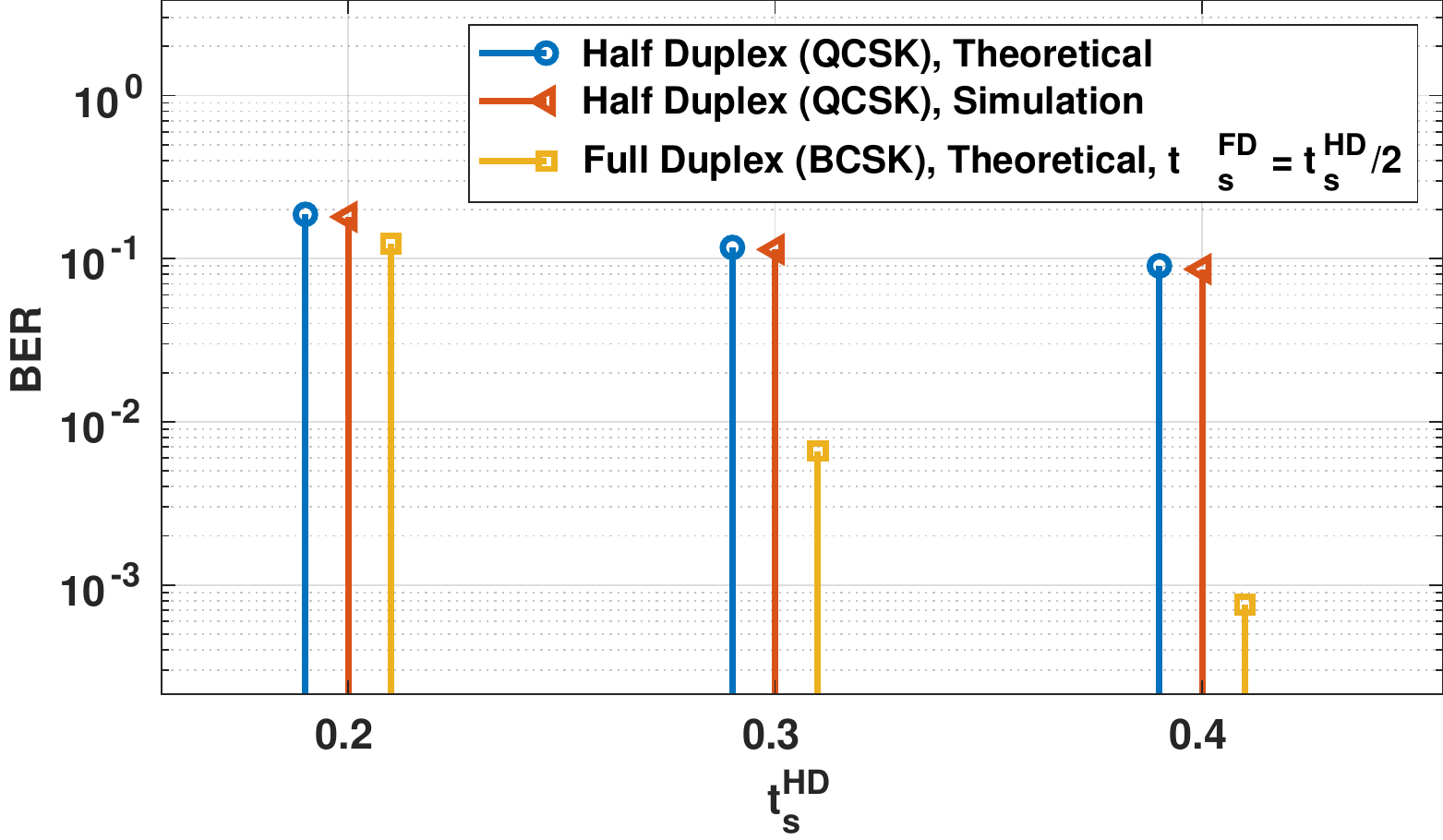}
		\caption{Theoretical and simulation BER of the half-duplex system (QCSK) and the theoretical BER of the full-duplex system with optimized D-SIC and A-SIC (BCSK), $\Ntx$ is 500.}
		\label{fig_qcsk_bcsk}
	\end{center}
\end{figure}

For the case i), Fig.~\ref{fig_bcsk_bcsk_ts_half} shows that the throughputs of the full-duplex system with optimized SIC are almost double the throughput of the half-duplex system. Thus, we can achieve nearly double the transmission rate using the proposed SIC techniques without degrading the BER significantly. When $t^{\text{HD}}_{s}\!=\!t^{\text{FD}}_{s}$ is $\SI{0.100}{\second}$ or $\SI{0.200}{\second}$, the BER of the full-duplex system with optimized SIC is less than the BER of the half-duplex system. The converse is true when $t^{\text{HD}}_{s}\!=\!t^{\text{FD}}_{s}$ is $\SI{0.300}{\second}$ or $\SI{0.400}{\second}$. Adjusted $t^{\text{FD}}_{s}$ throughput values in Table~\ref{table_bcsk_bcsk} show the same tendency. 
The overall results imply that the full-duplex system with optimized SIC becomes better than the half-duplex system even in terms of BER when $t_s$ get decreased.

\begin{table}
	\caption{Throughput ratio (FD/HD) of the half-duplex system using QCSK and the full-duplex system using BCSK with optimized D-SIC and A-SIC }
	\label{table_qcsk_bcsk}
	\begin{center}
		\begin{tabular}{c|ccc}
			\hline\hline 
			\hphantom{123} $\Ntx$ \phantom{$1^{2^{3}}$}  			& $t^{\text{HD}}_s \,(\si{\second})$ & $t^{\text{FD}}_s \,(\si{\second})$ & Throughput Ratio (FD/HD)\\
			\hline
			\multirow{ 3}{*}{300} & 0.200 			& 0.100 			& 1.222 \\
			& 0.300 			& 0.150 			& 1.295 \\
			& 0.400 			& 0.200 			& 1.274 \\
			\hline
			\multirow{ 3}{*}{400} & 0.200 			& 0.100 			& 1.134 \\
			& 0.300 			& 0.150 			& 1.188 \\
			& 0.400 			& 0.200 			& 1.162	 \\
			\hline
			\multirow{ 3}{*}{500} & 0.200 			& 0.100 			& 1.078 \\
			& 0.300 			& 0.150 			& 1.125 \\
			& 0.400 			& 0.200 			& 1.098 \\
			\hline\hline
		\end{tabular}
	\end{center}
\end{table}
In Fig.~\ref{fig_qcsk_bcsk} we depict the simulation and theoretical BER of the half-duplex system using QCSK and compare them to the full-duplex system with optimized SIC using BCSK where $t^{\text{FD}}_s\!=\!t^{\text{HD}}_s/2$. In this case, $M/t_s$ is the same for both systems. Hence, the BER determines the difference between the throughputs. For QCSK, we used an equally spaced number of molecules for encoding different bits (i.e., bit-0, 1, 2, 3) and three thresholds (i.e., $\tau_{d1}$, $\tau_{d2}$, $\tau_{d3}$) to detect the molecular received signal. Fig.~\ref{fig_qcsk_bcsk} shows that the BER of the half-duplex system using QCSK is much higher than that of the full-duplex system using BCSK with optimized SIC. The theoretical BERs of the half-duplex system using QCSK are calculated by using a straightforward extension of \eqref{eq_ber}. The throughput difference between the two systems can be seen in Table~\ref{table_qcsk_bcsk}.


\section{Conclusion}
\label{sec_conclusion}
In this paper, we have investigated two different communication models of two-way MCvD---a half-duplex system and a full-duplex system. We derived the multi-receiver channel model. We also derived the BER formula and verified the formula by simulation.
Theoretical analysis and simulations showed that severe SI occurs in the full-duplex system.  Therefore, we proposed two SIC techniques to mitigate this interference: A-SIC and D-SIC. 
We analyzed the BER improvements in the full-duplex system with the proposed SIC techniques and numerically found the optimal values for the normalized detection threshold and the discarding time in order to minimize the system BER.
To compare the half-duplex system with the full-duplex system, we evaluated the throughput in three different cases. The throughput of the full-duplex system with optimized SIC increased to more than that of the half-duplex system as $t_s$ decreased. With the proposed SIC techniques, we showed the possibility of full-duplex molecular communication using a single type of molecule.
On the other hand, the BER analysis and simulation results revealed that using a concentration-based modulation technique of higher order significantly degrades the BER. Investigating a more effective modulation technique for the two-way MCvD will be a topic for the future work. We will also study the channel model for more general system configurations.

\section{Appendix}
\subsection{Derivation of $d_{\fixpnew{1}{\ell}}^{\rxJ{2}}$ and $d_{\fixpnew{2}{\ell}}^{\rxJ{1}}$}
Using~\eqref{eqn_3d_frac_received_point_src}, we rewrite~\eqref{cdf_limit} and combine with~\eqref{eq_capture_prob} as
\begin{align}
\begin{split}
k_1&=\lim_{t\to\infty}\cdft{1}{\txJ{1}}{t} =\frac{\frac{r_{r1}}{r_{r1}\!+\!d_{\txJ{1}}^{\rxJ{1}}}-\frac{r_{r1}}{r_{r1}\!+\!d_{\fixpnew{2}{\ell}}^{\rxJ{1}}}\frac{r_{r2}}{r_{r2}\!+\!d_{\txJ{1}}^{\rxJ{2}}}}{1-\frac{r_{r2}}{r_{r2}\!+\!d_{\fixpnew{1}{\ell}}^{\rxJ{2}}}\frac{r_{r1}}{r_{r1}\!+\!d_{\fixpnew{2}{\ell}}^{\rxJ{1}}}}\\
k_2&=\lim_{t\to\infty}\cdft{2}{\txJ{1}}{t} =\frac{\frac{r_{r2}}{r_{r2}\!+\!d_{\txJ{1}}^{\rxJ{2}}}-\frac{r_{r2}}{r_{r2}\!+\!d_{\fixpnew{1}{\ell}}^{\rxJ{2}}}\frac{r_{r1}}{r_{r1}\!+\!d_{\txJ{1}}^{\rxJ{1}}}}{1-\frac{r_{r2}}{r_{r2}\!+\!d_{\fixpnew{1}{\ell}}^{\rxJ{2}}}\frac{r_{r1}}{r_{r1}\!+\!d_{\fixpnew{2}{\ell}}^{\rxJ{1}}}},
\end{split} 
\label{eq_appen_a1}
\end{align}
where $k_1$ and $k_2$ indicate $k_1(\mu,\eta,\phi)$ and $k_1(\mu,\eta,\phi)$, respectively.
Note that~\eqref{eq_appen_a1} can be rewritten as quadratic equations of $d_{\fixpnew{1}{\ell}}^{\rxJ{2}}$ and $d_{\fixpnew{2}{\ell}}^{\rxJ{1}}$, and the solutions of the equations are derived as the following:
\begin{align}
\begin{split}
d_{\fixpnew{1}{\ell}}^{\rxJ{2}}&=\frac{r_{r2}^2(k_1+k_2-1)+r_{r2}d_{\txJ{1}}^{\rxJ{2}}(k_1+k_2)}{r_{r2}(1-k_2)-d_{\txJ{1}}^{\rxJ{2}}k_2}\\
d_{\fixpnew{2}{\ell}}^{\rxJ{1}}&=\frac{r_{r1}^2(k_1+k_2-1)+r_{r1}d_{\txJ{1}}^{\rxJ{1}}(k_1+k_2)}{r_{r1}(1-k_1)-d_{\txJ{1}}^{\rxJ{1}}k_1}.
\end{split} 
\label{eq_appen_a2}
\end{align}
\label{appena}
\subsection{Derivation of \eqref{app_final}}
\label{appenb}
Through performing Laplace transform in \eqref{cdf_1_app} and solving simultaneous equations, we get
\begin{align}
\begin{split}
\cdft{2}{\txJ{1}}{s}=\frac{-(A-1)c_1 e^{-b_1\sqrt{\frac{s}{D}}}}{s(e^{B\sqrt{\frac{s}{D}}}-A)}+ \frac{(A-1)c_2 e^{-b_2\sqrt{\frac{s}{D}}}}{s(e^{B\sqrt{\frac{s}{D}}}-A)} \\
\cdft{1}{\txJ{1}}{s}=\frac{-(A-1)c_3 e^{-b_3\sqrt{\frac{s}{D}}}}{s(e^{B\sqrt{\frac{s}{D}}}-A)}+ \frac{(A-1)c_4 e^{-b_4\sqrt{\frac{s}{D}}}}{s(e^{B\sqrt{\frac{s}{D}}}-A)},
\end{split}
\label{cdf_laplace}
\end{align}
where $\cdft{2}{\txJ{1}}{s}$ and $\cdft{1}{\txJ{1}}{s}$ are the Laplace transform of $\cdft{2}{\txJ{1}}{t}$ and $\cdft{1}{\txJ{1}}{t}$, respectively. To derive the inverse Laplace transforms of \eqref{cdf_laplace}, we first derive the inverse Laplace transforms of $\cdft{2}{\txJ{1}}{s^2}$ and $\cdft{1}{\txJ{1}}{s^2}$, where 
\begin{align}
\begin{split}
\cdft{2}{\txJ{1}}{s^2}=\frac{-(A-1)c_1 e^{\frac{-b_1s}{\sqrt{D}}}}{s^2(e^{\frac{Bs}{\sqrt{D}}}-A)}+ \frac{(A-1)c_2 e^{\frac{-b_2s}{\sqrt{D}}}}{s^2(e^{\frac{Bs}{\sqrt{D}}}-A)}\\
\cdft{1}{\txJ{1}}{s^2}=\frac{-(A-1)c_3 e^{\frac{-b_3s}{\sqrt{D}}}}{s^2(e^{\frac{Bs}{\sqrt{D}}}-A)}+ \frac{(A-1)c_4 e^{\frac{-b_4s}{\sqrt{D}}}}{s^2(e^{\frac{Bs}{\sqrt{D}}}-A)}.
\end{split}
\label{cdf_laplace_sq}
\end{align}
Using~\cite[eqs. (5.1), (5.57), (1.18)]{laplace}, the inverse laplace transforms of \eqref{cdf_laplace_sq} can be expressed as the following:
\begin{align}
\begin{split}
L^{-1}\{\cdft{2}{\txJ{1}}{s^2}\}=k_3(t)-k_4(t) \\
L^{-1}\{\cdft{1}{\txJ{1}}{s^2}\}=k_1(t)-k_2(t),
\end{split}
\end{align}
where
\begin{align}
\begin{split}
k_i(t)&\!\!=\!u(t\!\!-\!\!b_i)\left\{\frac{A^n\!\!-\!\!1}{A\!\!-\!\!1}(t\!\!-\!\!b_i)\!+\!\frac{B(-nA^{n+1}\!\!+\!(n\!\!+\!\!1)A^n\!\!-\!\!1)}{(A\!\!-\!\!1)^2}\!\!\right\}, \\&b_i+(n+1)B<t<b_i+nB, n=-1,-2...-\infty.
\end{split}
\label{cdf_ilaplace_sq}
\end{align}
In \eqref{cdf_ilaplace_sq}, $u(t)$ is the unit step function. From \eqref{cdf_ilaplace_sq}, we derive the inverse Laplace transform of \eqref{cdf_laplace} as the following~\cite[eqs. (1.27)]{laplace}:
\begin{align}
\begin{split}
\cdft{2}{\txJ{1}}{t}=\frac{1}{2\sqrt{\pi t^3}} \int_{0}^{\infty}ue^{\frac{-u^2}{4t}}\left(L^{-1}\{\cdft{2}{\txJ{1}}{s^2}\}\right)du\\
\cdft{1}{\txJ{1}}{t}=\frac{1}{2\sqrt{\pi t^3}} \int_{0}^{\infty}ue^{\frac{-u^2}{4t}}\left(L^{-1}\{\cdft{1}{\txJ{1}}{s^2}\}\right)du,
\end{split}
\label{cdf_ilaplace}
\end{align}
which give us \eqref{app_final}.

\bibliographystyle{IEEEtran}
\bibliography{IEEEabrv,two_wayb}

%




\end{document}